\documentclass[preprint,9pt]{elsarticle}
\makeatletter
\def\ps@pprintTitle{%
 \let\@oddhead\@empty
 \let\@evenhead\@empty
 \def\@oddfoot{\centerline{\thepage}}%
 \let\@evenfoot\@oddfoot}
\makeatother
\usepackage{geometry}
\usepackage{floatrow}
\usepackage{subfig}
\usepackage{caption}

\usepackage[colorlinks=true,citecolor=blue]{hyperref}
\usepackage[dvips]{epsfig}
\usepackage{color}
\usepackage{bm}
\usepackage{amssymb}
\usepackage{amsmath}
\usepackage{empheq}
\allowdisplaybreaks[4]
\usepackage{amsthm}
\usepackage{etoolbox}
\makeatletter
\patchcmd{\@endtheorem}{\@endpefalse}{}{}{}
\patchcmd{\endproof}{\@endpefalse}{}{}{}
\makeatother
\usepackage{amsfonts}
\usepackage{xspace}
\usepackage{upgreek}
\usepackage{graphicx}
\usepackage{graphicx}
\usepackage{epstopdf}
\usepackage[table]{xcolor}
\usepackage[title]{appendix}
\usepackage{float}
\usepackage{dsfont}
\usepackage{hyperref}

\usepackage{amssymb, mathtools}
\usepackage[table]{xcolor}

\theoremstyle{plain}
\newtheorem{theorem}{Theorem}

\newtheorem{proposition}[theorem]{Proposition}
\newtheorem{corollary}[theorem]{Corollary}

\theoremstyle{definition}
\newtheorem{definition}[theorem]{Definition}

\theoremstyle{remark}

\biboptions{sort&compress}
\usepackage{siunitx}
\usepackage{setspace}
\usepackage{esvect}
\usepackage{url}

\DeclareSymbolFont{bbold}{U}{bbold}{m}{n}
\DeclareSymbolFontAlphabet{\mathbbold}{bbold}

\numberwithin{equation}{section}

\newcommand{\norm}[1]{\lVert #1 \rVert}
\newcommand{\Norm}[1]{\left\lVert #1 \right\rVert}

\usepackage{accents}

\usepackage{mathtools}

\usepackage{color}

\definecolor{forestgreen}{rgb}{0.33,0.61,0.34}

\newcommand{\del}[1]{\empty}

\begin{document}

\begin{frontmatter}



\title{Effects of concurrency on epidemic spreading in Markovian temporal networks}


\author[label1]{Ruodan Liu}
\author[label2]{Masaki Ogura}
\author[label1]{Elohim Fonseca Dos Reis}
\author[label1]{Naoki Masuda\corref{cor1}\fnref{label3,label4}}
\ead{naokimas@buffalo.edu}
\cortext[cor1]{Corresponding author}
 \address[label1]{Department of Mathematics, State University of New York at Buffalo, Buffalo, NY 14260-2900, USA}
 \address[label2]{Graduate School of Information Science and Technology, Osaka University, Suita, Osaka 565-0871, Japan}
 \address[label3]{Computational and Data-Enabled Sciences and Engineering Program, State University of New York at Buffalo, Buffalo, NY 14260-5030, USA}
 \address[label4]{Faculty of Science and Engineering, Waseda University, 169-8555 Tokyo, Japan}

\begin{abstract}
The concurrency of edges, quantified by the number of edges that share a common node at a given time point, may be an important determinant of epidemic processes in temporal networks. We propose theoretically tractable Markovian temporal network models in which each edge flips between the active and inactive states in continuous time. The different models have different amounts of concurrency while we can tune the models to share the same statistics of edge activation and deactivation (and hence the fraction of time for which each edge is active) and the structure of the aggregate (i.e., static) network. We analytically calculate the amount of concurrency of edges sharing a node for each model. We then numerically
study effects of concurrency on epidemic spreading in the stochastic susceptible-infectious-susceptible  and susceptible-infectious-recovered dynamics on the proposed temporal network models. We find that the concurrency enhances epidemic spreading near the epidemic threshold while this effect is small in many cases. Furthermore, when the infection rate is substantially larger than the epidemic threshold, the concurrency suppresses epidemic spreading in a majority of cases. In sum, our numerical simulations suggest that the impact of concurrency on enhancing epidemic spreading within our model is consistently present near the epidemic threshold but modest. The proposed temporal network models are expected to be useful for investigating effects of concurrency on various collective dynamics on networks including both infectious and other dynamics.
\end{abstract}

\begin{keyword}
Concurrency, temporal network, SIS model, SIR model, epidemic threshold, Poisson process.
\end{keyword}

\end{frontmatter}

\section{Introduction}

The structure of contact networks among individuals shapes the dynamics of contagion processes in a population such as the number of individuals infected, their spatial distributions, and speed of spreading \cite{Newman2018book, Pastor2015RMP, Barrat2008book, Kiss2017book}. For example, a heterogeneous degree distribution, where the degree is the number of neighboring nodes that a node has, and a short average distance between nodes are two factors that usually enhance epidemic spreading on networks. In fact, empirical contact networks often vary over time on a time scale comparable to or faster than that of epidemic dynamics, probably most famously owing to mobility of human or animal individuals. This observation has naturally led to the investigation of how features of such temporal (i.e., time-varying) networks and statistical properties of contact events, such as distributions of inter-contact times, temporal correlation in inter-contact times, and appearance and disappearance of nodes and edges, affect outcomes of contagion processes~\cite{Holme2012PR, Holme2015EPJB, Masuda2020book, Bansal2010JBD, Masuda2013F1000PR, Moody2002SF}.

The concurrency is a feature of temporal contact networks that has been investigated in both mathematical and field epidemiology for over two decades, while many of these studies do not specifically refer to temporal networks \cite{Watts, Moody, Morris1995, Kretz1996, Morris1997, Goodreau, Masuda2021RoS, Lee, Morris2010PlosONE}. Concurrency generally refers to multiple partnerships of an individual that overlap in time. Concurrency has been examined in particular in the context of sexually transmitted infections such as HIV, specifically regarding whether or not the concurrency was high in sub-Saharan Africa and, if so, whether or not the high concurrency enhanced HIV spreading there \cite{Moody, Morris1995, Kretz1996, Aral2010CIDR, Kretzschmar2012AIDS, Lurie2010AIDS, Foxman2006STD}. Suppose that a partnership between individuals $i$ and $j$ overlaps one between $i$ and $j^\prime$ for some time.
Such concurrency may promote epidemic spreading because 
rapid transmission from $j$ to $j^\prime$ through $i$ and vice versa is possible during the concurrent partnerships (see Fig.~\ref{fig:3nodes}(a)). In contrast, if a partnership between $i$ and $j$ occurs before that between $i$ and $j'$ without an overlap, i.e., without concurrency, the transmission from $j$ to $j'$ can still occur in various different ways, but the transmission from $j'$ to $j$ does not (see Fig.~\ref{fig:3nodes}(b)).

Modeling studies are diverse in how to quantify the amount of concurrency. Many studies measure the concurrency through the mean degree, or the average contact rate for nodes \cite{Watts, Doherty, Gurski, Eames}. 
%
%
However, in these cases, it is hard to say whether an increased epidemic size owes to a high density of edges or a large amount of concurrency \cite{Onaga2017PRL, Miller, Masuda2021RoS, Bauch2000PRSB}. In studies of epidemic processes on static networks, it is a stylized fact that epidemic spreading is enhanced if there are many edges with all other things being equal. Another line of theoretical and computational approach to concurrency quantifies concurrency by the heterogeneity in the degree distribution of the network \cite{Morris1995, Kretz1996, Masuda2021RoS}. However, the results that a high concurrency in this sense enhances epidemic spreading is equivalent to an established result in static network epidemiology that heterogeneous degree distributions lead to a small epidemic threshold, thus promoting epidemic spreading \cite{Pastor2015RMP, Barrat2008book, Pastor2001PRL, Kiss2017book}. In some studies, the authors investigated the effect of concurrency by carefully ensuring that the degree distribution (and hence the average degree) stays the same across the comparisons while they manipulate the amount of concurrency \cite{Miller, Onaga2017PRL, Bauch2000PRSB}. These studies suggest that, despite the different definitions of the concurrency employed in these studies, higher concurrency increases epidemic spreading (but see \cite{Miller} for different results). However, mathematical models that enable us to analyze the effect of concurrency by fixing the structure of the static network are still scarce. 

\floatsetup[figure]{style=plain,subcapbesideposition=top}
\captionsetup{font={small,rm}} 
\captionsetup{labelfont=bf}
\begin{figure}[H]
  \centering
  \includegraphics[width=.8\linewidth]{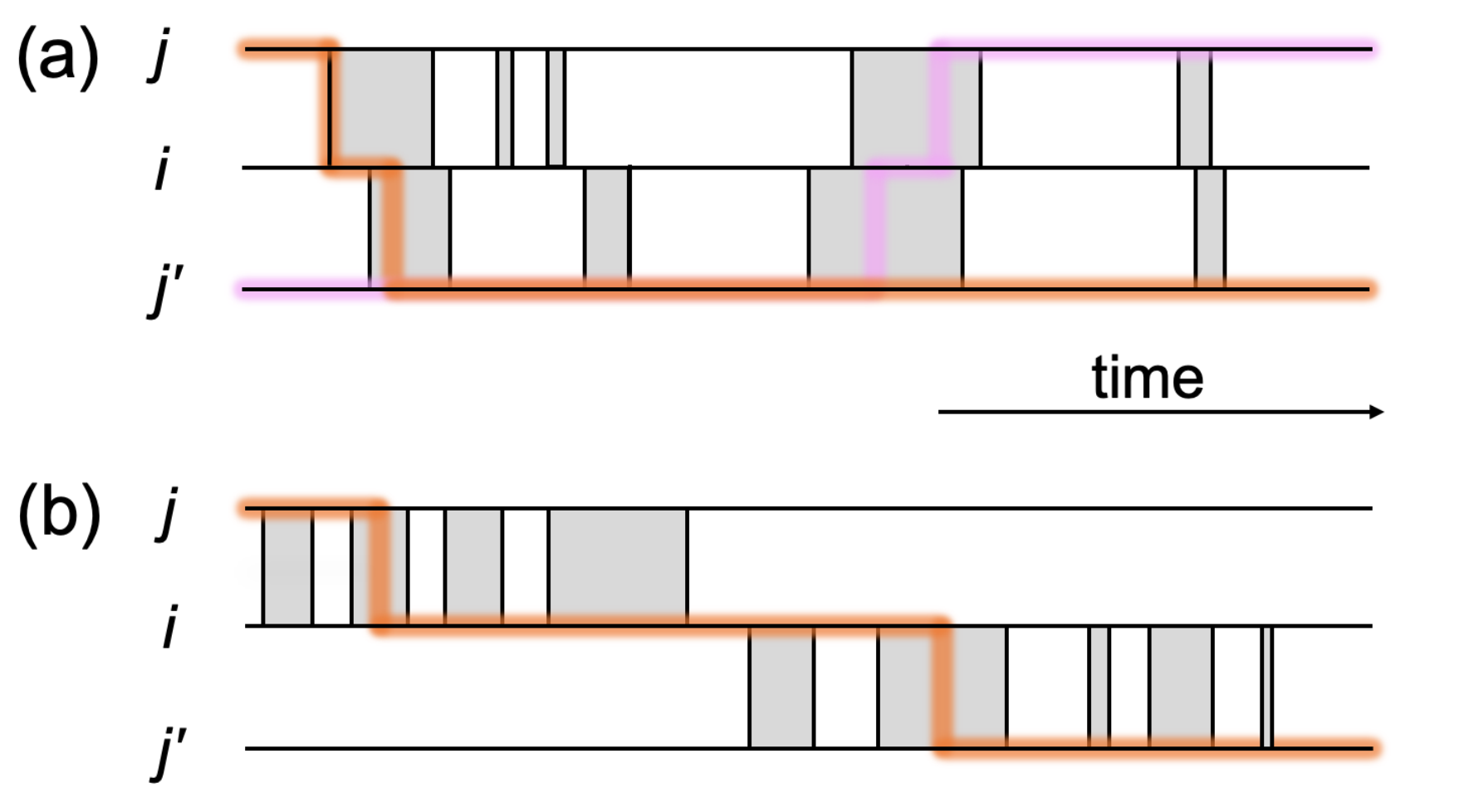}
   \caption{Schematic of part of temporal networks with different amounts of concurrency. We depict two edges sharing a node in each case. (a) Concurrent partnerships. (b) Non-concurrent partnerships. A shaded box represents a duration for which a partnership is present on the edge. The thick lines are examples of time-respecting paths transmitting infection from $j$ to $j'$ (shown in orange) or from $j'$ to $j$ (shown in magenta).
This figure is inspired by Fig.~1 in Ref.~\cite{Masuda2021RoS} and Fig.~1 in Ref.~\cite{Miller}.}
   \label{fig:3nodes}
\end{figure}

In the present study, we focus on effects of concurrency on epidemic spreading for an arbitrary static network on top of which partnerships, infection events, and recovery events occur. In this manner, we aim to study the effect of concurrency without being affected by the effect of the network structure itself. We consider three temporal network models that have different amounts of concurrency while keeping the probability that each edge is available the same across comparisons. In terms of a concurrency measure, we analytically evaluate the amount of concurrency of edges sharing a node for the three models as well as other properties of the models. Then, we numerically study effects of concurrency on epidemic spreading using the stochastic susceptible-infectious-susceptible (SIS) and susceptible-infectious-recovered (SIR) models. 

The Python codes for generating the numerical results in this article are available at Github (\url{https://github.com/RuodanL/concurrency}).

\section{Temporal network models}

We introduce three models of undirected and unweighted temporal networks in continuous time, which we build from independent Markov processes occurring on a given static network. We use these models to compare temporal networks that have the same time-aggregated network but different amounts of concurrency. Let $G$ be an undirected and unweighted static network with node set $V = \left\{1, ... , N\right\}$ and edge set $E = \left\{e_i\right\}_{i=1}^{M}$, where $N$ is the number of nodes and $M$ is the number of edges. We assume that the network $G$ has no self-loops. In the temporal network constructed based on undirected and unweighted static network $G$, by definition, each edge $e_i \in E$ is either active (i.e., temporarily present) or inactive (i.e., temporarily absent) at any given time $t \in \mathbb{R}$ and switches between the active and inactive states over time. We denote by $\mathcal{G}$ the temporal network and by $\mathcal{G}(t)$ the instantaneous network in $\mathcal{G}$ observed at time $t$. 

\subsection{Model 1}

Let $\tau_1$ and $\tau_2$ be the duration of the active state of and the inactive state of an edge, respectively. 
In our model 1, we assume that $\tau_1$ is independently drawn from probability density $\psi_1(\tau_1)$ each time the edge switches from the inactive to the active state.
Similarly, we draw $\tau_2$ independently from probability density $\psi_2(\tau_2)$ each time the edge switches from the active to the inactive state.
The duration $\tau_1$ or $\tau_2$ for different edges obeys the same distributions but is drawn independently for the different edges.
When $\psi_1(\tau_1)$ and $\psi_2(\tau_2)$ are both exponential distributions, the stochastic dynamics of the state of each edge
obeys independent continuous-time Markov processes with two states, and our model 1 reduces to previously proposed models \cite{Clementi, Newman2}.

\subsection{Model 2}

In models 2 and 3, we assume that each node independently obeys a continuous-time Markov process with two states, which we refer to as the high-activity and low-activity states. Each node independently switches between the high-activity state, denoted by $h$, and the low-activity state, denoted by $\ell$. Let $a$ be the rate at which the state of a node changes from $\ell$ to $h$, and let $b$ be the rate at which the state of a node changes from $h$ to $\ell$. In model 2, each edge $(u, v) \in E$, where $u, v \in V$, is active at any given time if and only if both $u$ and $v$ are in the $h$ state. The dynamics of a single edge in model 2 is schematically shown in Fig.~\ref{fig:m2&3}. The model resembles the so-called AND model in \cite{Elohim}. An intuitive interpretation of model 2 is that two individuals chat with each other if and only if both of them want to.

\floatsetup[figure]{style=plain,subcapbesideposition=top}
\captionsetup{font={small,rm}} 
\captionsetup{labelfont=bf}
\begin{figure}[H]
  \centering
  \includegraphics[width=1.0\linewidth]{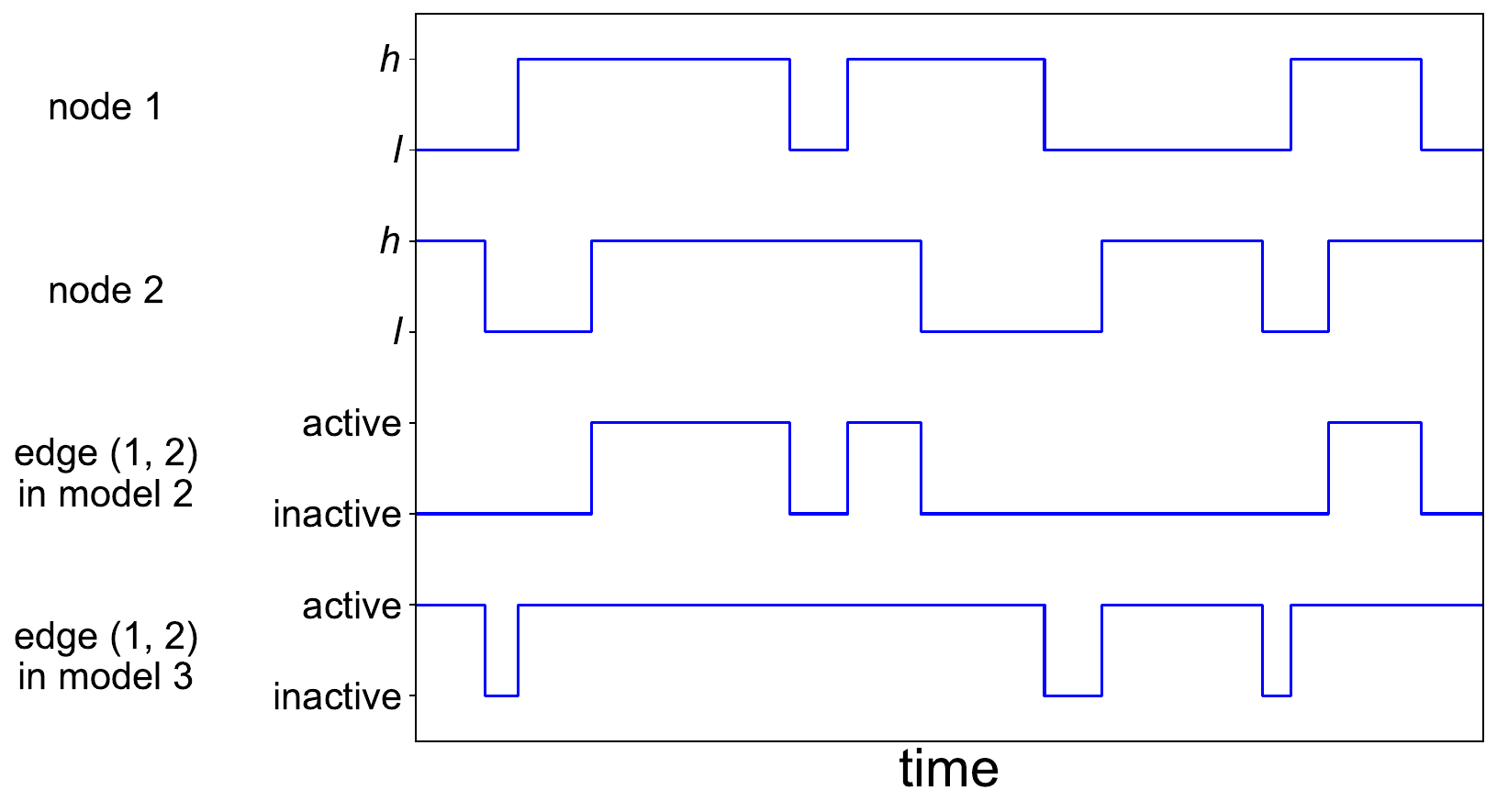}
   \caption{Schematic illustration of models 2 and 3. In model 2, edge $(1, 2)$ is active if and only if both node $1$ and node $2$ are in the $h$ state. Otherwise, the edge is inactive. In model 3, edge $(1, 2)$ is active if and only if either node $1$ or node $2$ is in the $h$ state.}
   \label{fig:m2&3}
\end{figure}

\subsection{Model 3}

Model 3 is a variant of model 2. As in model 2, we assume that each node independently switches between the high-activity and low-activity states at rates $a$ and $b$. In model 3, the edge between two nodes is active if and only if either node is in the $h$ state (see Fig.~\ref{fig:m2&3} for a schematic). This model resembles the OR model proposed in \cite{Elohim}. The intuition behind model 3 is that one person can start conversation with another person whenever either person wants to talk regardless of whether the other person wants to.

\section{Concurrency for the three models}

In this section, we first define a concurrency index for temporal networks in which partnership on each edge appears and disappears over time.
Then, we calculate and compare the concurrency index and related measures for models 1, 2, and 3.
 
\subsection{Definition of a concurrency index}

Consider a temporal network $\mathcal{G}$ that is constructed on the underlying static network $G$ and defined for time $t \in [0, T]$. We refer to $G$ as the aggregate network.
If two edges in $\mathcal{G}$ sharing a node in $G$ are both active at $t \in \mathbb{R}$, we say that the two edges are concurrent at time $t$; see Fig.~\ref{fig:3nodes}(a) for an example.
Consider a pair of edges sharing a node in the aggregate network, denoted by $e_i$ and $e_j$, where $e_i, e_j \in E$.
In fact, the likelihood that $e_i$ and $e_j$ are concurrent at time $t$ depends on how likely each single edge is active at time $t$.
We define a concurrency index that takes into account of this factor. To this end, we first define the set of time for which an edge $e \in E$ is active by
\begin{equation}
S(e) = \{ t \in [0, T]; \text{ edge } e \text{ is active at time } t \}.
\end{equation}
We define the concurrency for the edge pair $\{ e_i, e_j \} \in \mathcal{S}$ by
\begin{equation} \label{concurrency}
\kappa(e_i, e_j) = \frac{ m\left(S(e_i) \cap S(e_j)\right) }
{\min\{ m\left( S(e_i) \right), m\left( S(e_j) \right) \}},
\end{equation}
where $m$ denotes the Lebesgue measure, and $\mathcal{S}$ is the set of edge pairs that share a node in $G$. Note that
\begin{equation}
\left| \mathcal{S} \right|
= \sum_{i=1}^N \frac{\overline{k}_i (\overline{k}_i-1)}{2},
\label{eq:cardinality-mathcalS}
\end{equation}
where $\overline{k}_i$ is the degree of the $i$th node in $G$  \cite{Morris1995, Kretz1996}.
The numerator on the right-hand side of Eq.~\eqref{concurrency} is equal to the length of time for which both $e_i$ and $e_j$ are active.
The denominator is a normalization constant to discount the fact that the numerator would be large if the two edges are active for long time.
Because any edge $e$ in the aggregate network $G$ should be active sometime in $[0, T]$, the denominator is always positive; otherwise, we should exclude $e$ from $G$.

We define the concurrency index for temporal network $\mathcal{G}$ by
\begin{align}
\kappa(\mathcal{G}) & = \frac{1}{\left| \mathcal{S} \right|}
\sum_{i, j \text{ such that } 1\le i < j\le M \text{ and } \{ e_i, e_j \} \in \mathcal{S}} \kappa(e_i, e_j).
\end{align}
It holds true that $0 \le \kappa(\mathcal{G}) \le 1$ because $0 \le \kappa(e_i, e_j) \le 1$.
For empirical or numerical data, $[0, T]$ is the observation time window. For a stochastic temporal network model, we calculate $S(e)$ as the expectation and in the limit of $T\to\infty$ such that $\kappa(\mathcal{G})$ is a deterministic quantity.

Differently from $\kappa_3$, a concurrency index proposed in seminal studies \cite{Morris1995, Kretz1996} (also see for \cite{Masuda2021RoS} a review), 
$\kappa(\mathcal{G})$ is not affected by the degree distribution of the aggregate network $G$. It should be noted that the calculation of $\kappa(e_i, e_j)$ and hence $\kappa(\mathcal{G})$ requires the information about the aggregate network, i.e., $\mathcal{S}$. 

\subsection{Model 1}

To calculate the concurrency index for the three models of temporal networks, we first derive the probability of an arbitrary edge being active in the equilibrium for each model. In model 1, consider an arbitrary edge in the static network $G$. The mean duration for which the edge is active and that for which the edge is inactive are given by 
\begin{align}
\langle \tau_1\rangle & = \int_{0}^{\infty} \tau_1\psi_1(\tau_1) d\tau_1
\end{align}
and
\begin{align}
\langle \tau_2\rangle & = \int_{0}^{\infty} \tau_2\psi_2(\tau_2) d\tau_2,
\end{align}
respectively.
Owing to the renewal reward theorem \cite{Ross}, the probability that the edge is active in the equilibrium, denoted by $q^*$, is given by
\begin{equation}
q^*=\frac{\langle \tau_1\rangle}{\langle \tau_1\rangle+\langle \tau_2\rangle}.
\end{equation}
The concurrency index, $\kappa(\mathcal{G})$, which we simply refer to as $\kappa$ in the following text, is equal to the ratio of the time for which the two edges sharing a node are both active to the time for which an edge is active. Because the states of different edges are independent of each other, we obtain
\begin{equation} \label{model1}
\kappa=\frac{{q^{*}}^2}{q^*} = q^{*} = \frac{\langle \tau_1\rangle}{\langle \tau_1\rangle+\langle \tau_2\rangle}.
\end{equation}

The special case of model 1 in which the edge activation and deactivation occur as Poisson processes is equivalent to previously proposed 
models \cite{Newman2, Clementi}. In this case, using
$\psi_1(\tau_1)=\lambda_1e^{-\lambda_1\tau_1}$ and $\psi_2(\tau_2)=\lambda_2e^{-\lambda_2\tau_2}$, we obtain
\begin{align}
\kappa & = \frac{\langle \tau_1\rangle}{\langle \tau_1\rangle+\langle \tau_2\rangle}
= \frac{\lambda_1^{-1}}{\lambda_1^{-1}+\lambda_2^{-1}} = \frac{\lambda_2}{\lambda_1+\lambda_2}.
\end{align}

\subsection{Model 2}

To analyze model 2, let us consider a pair of neighboring nodes $v_1$ and $v_2$ in the static network $G$. Let $p^*_h$ and $p^*_\ell$ be the probability that an arbitrary node in $G$ is in state $h$ and $\ell$ in the equilibrium, respectively. Denote by $p^*_{s_1s_2}$ the probability that node $v_i$ is in state $s_i \in \{ h, \ell\}$ in the equilibrium, where $i \in \{1, 2\}$. Because the duration of the high-activity state and that of the low-activity state of a node obey the exponential distributions with mean $1/b$ and $1/a$, respectively, we apply the renewal reward theorem \cite{Ross} to obtain
\begin{equation}
p^*_h = \frac{\frac{1}{b}}{\frac{1}{a}+\frac{1}{b}} = \frac{a}{a+b}
\end{equation}
and
\begin{equation}
p^*_\ell = \frac{\frac{1}{a}}{\frac{1}{a}+\frac{1}{b}} = \frac{b}{a+b}.
\end{equation}
Because the states of different nodes are independent, we obtain
\begin{align} \label{equi1}
p_{hh}^{*} & = (p^*_h)^2 = \frac{a^2}{(a+b)^2}, \\
p_{h\ell}^{*} & = p_{\ell h}^{*} = p^*_hp^*_\ell = \frac{ab}{(a+b)^2},  
\end{align}
and
\begin{equation}
p_{\ell\ell}^{*} = (p^*_\ell)^2 = \frac{b^2}{(a+b)^2}.
\end{equation}
Therefore, the probability that an edge is active in the equilibrium is given by
\begin{equation}\label{m2q^*}
q^* = p_{hh}^{*} = \frac{a^2}{(a+b)^2}.
\end{equation}
To derive the concurrency for model 2, we consider a pair of edges sharing a node in $G$, denoted by $(v_1, v_2)$ and $(v_2, v_3)$. Denote by $p_{s_1s_2s_3}$ the probability that node $v_i$ is in state $s_i \in \{ h, \ell\}$, where $i=1$, $2$, and $3$. Similar to the analysis of $q^*$ for model 2, we obtain the following stationary probabilities:
\begin{align} \label{equi2}
p_{hhh}^{*} & = \frac{a^3}{(a+b)^3}, \\
p_{hh\ell}^{*} & = p_{h\ell h}^{*} = p_{\ell hh}^{*} = \frac{a^2b}{(a+b)^3},  \\
p_{h\ell\ell}^{*} & = p_{\ell h\ell}^{*} = p_{\ell\ell h}^{*} = \frac{ab^2}{(a+b)^3}, \\
p_{\ell\ell\ell}^{*} & = \frac{b^3}{(a+b)^3}.
\end{align}
Therefore, we obtain
\begin{equation} \label{model2}
\kappa=\frac{p_{hhh}^{*}}{q^*} = \frac{a}{a+b} = \sqrt{q^*}.
\end{equation}

\subsection{Model 3}

Similarly, for model 3, we obtain
\begin{equation}\label{m3q^*}
q^* = p_{hh}^{*} + p_{h\ell}^{*} + p_{\ell h}^{*} = \frac{a(a+2b)}{(a+b)^2},
\end{equation}
and
\begin{equation} \label{model3}
\kappa=\frac{p_{hhh}^{*}+p_{hh\ell}^{*}+p_{h\ell h}^{*}+p_{\ell hh}^{*}+p_{\ell h\ell}^{*}}{q^*} = \frac{a^2+3ab+b^2}{(a+b)(a+2b)}=\frac{q^*+\sqrt{1-q^*}}{1+\sqrt{1-q^*}}.
\end{equation}

\subsection{Comparison among the three models}

A large fluctuation of $\mathcal{G}(t)$ over time may impact concurrency \cite{Masuda2021RoS}.
In this section, we compare the amount of concurrency between the three models. 
\begin{proposition}
Model 2 is more concurrent than model 1 given that $q^*$ is the same between the two models.
\end{proposition}
\begin{proof}
Because $q^*$ is the same between the two models and $0<q^*<1$, we obtain $\sqrt{q^*} > q^*$ using Eqs.~\eqref{model1} and \eqref{model2}.
Therefore, model 2 is more concurrent than model 1.
\end{proof}
\begin{proposition}
Model 3 is more concurrent than model 1 given that $q^*$ is the same between the two models.
\end{proposition}
\begin{proof}
Because $q^*$ is the same between the two models and $0<q^*<1$, we obtain
\begin{align}
\frac{q^*+\sqrt{1-q^*}}{1+\sqrt{1-q^*}} - q^* = \frac{(1-q^*)^{\frac{3}{2}}}{1+(1-q^*)^{\frac{1}{2}}} > 0.
\end{align}
Therefore, model 3 is more concurrent than model 1.
\end{proof}
\begin{proposition}\label{cc23}
Given that $q^*$ is the same between models 2 and 3, \\
\begin{enumerate}[(i)]
\item model 3 is more concurrent than model 2 if $0<q^*<\frac{1}{2}$,
\item model 3 is less concurrent than model 2 if $\frac{1}{2}<q^*\leq 1$,
\item model 3 is equally concurrent to model 2 if $q^*=\frac{1}{2}$.
\end{enumerate}
\end{proposition}
\begin{proof}
For the sake of the present proof, let $a_2$ and $a_3$ be the rate at which the state of a node changes from $\ell$ to $h$ for models 2 and 3, respectively. Likewise, let $b_2$ and $b_3$ be the rate at which the state of a node changes from $h$ to $\ell$ for models 2 and 3, respectively. By imposing that $q^*$ is the same between the two models, we obtain
\begin{equation} \label{eo23}
q^* = \frac{a_3(a_3+2b_3)}{(a_3+b_3)^2} = \frac{{a_2}^2}{\left(a_2+b_2\right)^2}.
\end{equation}
Therefore, the difference in the concurrency between the two models, given by Eqs.~\eqref{model2} and \eqref{model3}, is given by
\begin{align}
\frac{a_3^2+3a_3b_3+b_3^2}{(a_3+b_3)(a_3+2b_3)} - \frac{a_2}{a_2+b_2} & = 
\frac{a_3^2+3a_3b_3+b_3^2}{(a_3+b_3)(a_3+2b_3)} - \frac{\sqrt{a_3(a_3+2b_3)}}{a_3+b_3} \nonumber\\
%
%
%
& = \frac{b_3^2\left[\left(b_3-a_3\right)^2-2a_3^2\right]}{\left(a_3+b_3\right)\left(a_3+2b_3\right)\left[a_3^2+3a_3b_3+b_3^2+(a_3+2b_3)\sqrt{a_3(a_3+b_3)}\right]}.
\end{align}
Therefore, the concurrency of model 3 is larger than that of model 2 if and only if $b_3>(1+\sqrt{2})a_3$, which is equivalent to $q^* = \frac{a_3(a_3+2b_3)}{(a_3+b_3)^2} = \frac{1+2\frac{b_3}{a_3}}{\left(1+\frac{b_3}{a_3}\right)^2} < \frac{1}{2}$.
\end{proof}

Using Eqs.~\eqref{model1}, \eqref{model2}, and \eqref{model3}, we compare the amount of concurrency for models 1, 2, and 3 in Fig. \ref{fig:concom}. We find that, when $q^*>1/2$, the concurrency index for model 2 is only slightly larger than that for model 3.

\floatsetup[figure]{style=plain,subcapbesideposition=top}
\captionsetup{font={small,rm}} 
\captionsetup{labelfont=bf}
\begin{figure}[H]
  \centering
  \includegraphics[width=.7\linewidth]{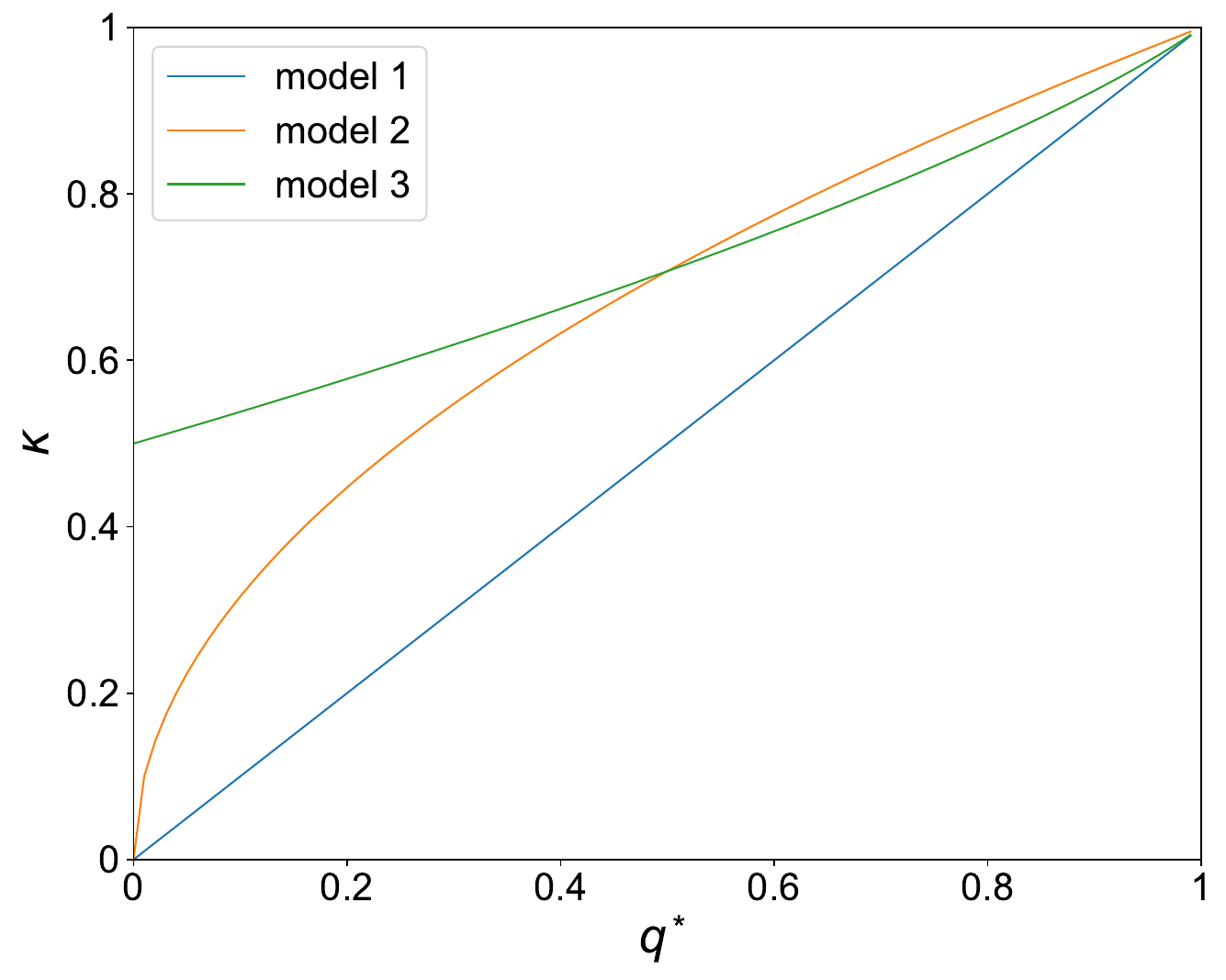}
   \caption{Concurrency index, $\kappa$, as a function of the stationary probability that an edge is active, $q^*$, for models 1, 2, and 3.}
   \label{fig:concom}
\end{figure}

\subsection{Fluctuations in the node's degree}

Consider the degree of individual nodes or its average over all nodes at any given time $t$. 
Let us fix their expectation and compare their statistical fluctuations, such as the standard deviation, across different models.
In a model with a higher concurrency, edges sharing a node in the aggregate network are more likely to be simultaneously active or inactive, which makes
the statistical fluctuation of the degree larger. Therefore, we anticipate that the concurrency of a temporal network model and the statistical fluctuation in the node's degree are interrelated. In this section, we establish such relationships for each of the three models. 
Let $A=(A_{ij})_{N \times N}$ be the adjacency matrix of static network $G$. 

To analyze the fluctuation in the node's degree in model 1, we define random variables by
\begin{equation}
x_{ij}(t) = 
\begin{cases}
0 & \text{if edge $(i, j)$ is inactive at time $t$}, \\
1 & \text{if edge $(i, j)$ is active at time $t$}.
\end{cases}
\end{equation}
The adjacency matrix of the temporal network at time $t$, $\mathcal{G}(t)$, of model 1 is given by the $N\times N$ matrix $B(t)=(B_{ij}(t))$, where $B_{ij}(t)=A_{ij}x_{ij}(t)$ (with $i, j \in \{1, \ldots, N\}$).
Let $k_i(t) \equiv \sum_{j=1}^N B_{ij}(t)$ be the degree of the $i$th node in network $\mathcal{G}(t)$. The average degree at time $t$ is given by
\begin{equation} \label{average_degree}
\langle k\rangle(t) \equiv \frac{1}{N} \sum_{i=1}^N k_i(t).
\end{equation}
In the following text, we omit $t$ because we discuss the fluctuations in $k_i(t)$ and $\langle k \rangle(t)$ in the equilibrium.
\begin{proposition} \label{1}
For model 1, it holds true that $k_i\sim B(\overline{k}_i, q^*)$ and $\langle k\rangle\sim\frac{2}{N}B\left(M, q^*\right)$, where $B(\cdot, \cdot)$ represents the binomial distribution. (We remind that $\overline{k}_i$ is the degree of the $i$th node in $G$ and that $M$ is the number of edges in $G$.)
\end{proposition}
\begin{proof}
For a proposition $C$, define the indicator function by
\begin{equation}
\mathds{1}(C) = \begin{cases}
1 &\text{if } C \text{ is true},\\
0 &\text{if } C \text{ is false}.
\end{cases}
\end{equation}
For a node $i$, we obtain
\begin{align}
k_i = \sum_{j=1}^N B_{ij} = \sum_{j=1}^N \mathds{1}(A_{ij}=1) x_{ij}.
\end{align}
Because $x_{ij}$ are independent and identically distributed Bernoulli random variables and $k_i$ is the sum of $\overline{k}_i$ terms, $k_i$ obeys $B\left(\overline{k}_i, q^*\right)$.          

Because we have assumed that the network is undirected, we obtain
\begin{align}
\langle k\rangle & = \frac{1}{N} \sum_{i=1}^N k_i\nonumber\\
& = \frac{2}{N} \sum_{i=1}^N \sum_{j=i+1}^N B_{ij}\nonumber\\
& = \frac{2}{N} \sum_{i=1}^N \sum_{j=i+1}^N \mathds{1}(A_{ij}=1)x_{ij}.
\end{align}
Because there are $M$ terms that comprise the summation, $\langle k\rangle$ obeys $\frac{2}{N}B\left(M, q^*\right)$.             
\end{proof}
%
%
We denote the expectation by $\mathbb E$ and the variance by $\sigma^2$; the standard deviation is equal to $\sigma$.
Using Proposition~\ref{1}, we obtain
\begin{align}
\mathbb E[k_i] & = \overline{k}_iq^*,\label{Ei1}\\
\sigma^2[k_i] & = \overline{k}_iq^*(1-q^*),\label{si1}\\
\mathbb E[\langle k\rangle] & = \frac{2Mq^*}{N},\label{E1}\\
\sigma^2[\langle k\rangle] & = \frac{4Mq^*(1-q^*)}{N^2}\label{s1}.
\end{align}

To analyze models 2 and 3, we define 
\begin{equation}
y_{i}(t) = 
\begin{cases}
0 & \text{if node $i$ is in the $\ell$ state at time $t$}, \\
1 & \text{if node $i$ is in the $h$ state at time $t$}.
\end{cases}
\end{equation} 
The adjacency matrix of $\mathcal{G}(t)$ for model 2 is given by $B(t)=(B_{ij}(t))$, where $B_{ij}(t)=A_{ij}y_{i}(t)y_{j}(t)$.
Using this expression, we can prove the following proposition.
\begin{proposition} \label{2}
For model 2, we obtain
\begin{align}
\mathbb E[k_i] =& \overline{k}_i\left(\frac{a}{a+b}\right)^2,\label{Ei2}\\
\sigma^2[k_i] =& \frac{\overline{k}_ia^2b}{(a+b)^3}\left(1+\frac{\overline{k}_ia}{a+b}\right),\label{Di2}\\
\mathbb E[\langle k\rangle] =& \frac{2M}{N}\left(\frac{a}{a+b}\right)^2,\label{E2}\\
\sigma^2[\langle k\rangle] =& \frac{4M}{N^2}\left(\frac{a}{a+b}\right)^2 + \frac{4M(M-1)}{N^2}\left(\frac{a}{a+b}\right)^3 - \frac{8}{N^2}\left(\frac{a}{a+b}\right)^3\frac{b}{a+b}\sum_{i=1}^{N-1}\sum_{j=i+1}^N\sum\limits_{\substack{k>i\\k\neq j}}^{N-1}\sum\limits_{\substack{\ell>k\\\ell\neq i,j}}^N A_{ij}A_{k\ell}\nonumber\\ 
& \quad - \frac{4M^2}{N^2}\left(\frac{a}{a+b}\right)^4.\label{D2}
\end{align}
\end{proposition} 
We prove Proposition~\ref{2} in Appendix~\ref{Appendix: A}.

The adjacency matrix of $\mathcal{G}(t)$ for model 3 is given by $B(t)=(B_{ij}(t))$, where $B_{ij}=A_{ij}\left[1 - \left(1 - y_i(t)\right)\left(1 - y_j(t)\right)\right]$.
Using this expression, we can prove the following proposition.
\begin{proposition} \label{4}
For model 3, we obtain
\begin{align}
\mathbb E[k_i] =& \frac{\overline{k}_ia(a+2b)}{(a+b)^2},\label{Ei3}\\
\sigma^2[k_i] =& \frac{\overline{k}_ia(a+2b)}{(a+b)^2}+\frac{\overline{k}_i(\overline{k}_i-1)a(a^2+3ab+b^2)}{(a+b)^3}-\left[\frac{\overline{k}_ia(a+2b)}{(a+b)^2}\right]^2,\label{Di3}\\
\mathbb E[\langle k\rangle] =& \frac{2Ma(a+2b)}{N(a+b)^2},\label{E3}\\
\sigma^2[\langle k\rangle] =& \frac{4Ma(a+2b)}{N^2(a+b)^2}+\frac{4M(M-1)a(a^2+3ab+b^2)}{N^2(a+b)^3}\nonumber\\
& \quad \ - \frac{8ab^3}{N^2(a+b)^4}\sum_{i=1}^{N-1}\sum_{j=i+1}^N\sum\limits_{\substack{k>i\\k\neq j}}^{N-1}\sum\limits_{\substack{\ell>k\\\ell\neq i,j}}^N A_{ij}A_{k\ell} - \frac{4M^2a^2(a+2b)^2}{N^2(a+b)^4}.\label{D3}
\end{align}
\end{proposition} 
We prove Proposition~\ref{4} in Appendix~\ref{Appendix: B}.

Now, let us compare the variance of $k_i$ and $\langle k\rangle$ among the three models under the condition that the expectation of $k_i$ and $\langle k\rangle$ is the same across the different models. This condition is equivalent to keeping $q^*$ the same across the models for each edge. The purpose of examining the variance of the node's degree is the following. Similar to Eq.~\eqref{eq:cardinality-mathcalS}, the number of concurrent edge pairs at time $t$ is given by $\sum_{i=1}^N k_i (k_i-1)/2$. Because this expression contains the second moment of the degree, we expect that the variance of the degree may be related to our concurrency measure.

In fact, we obtain the following results for the fluctuation of the degree, which are parallel to those for the concurrency index.
\begin{proposition} \label{3}
Assume that $\mathbb E[k_i]$ is the same between models 1 and 2 for any $i\in \{1, \ldots, N\}$.
For any given $q^*$, the variance $\sigma^2[k_i]$ is larger for model 2 than model 1 if $\overline{k}_i>1$. Likewise, $\sigma^2[\langle k\rangle]$ is larger for model 2 than model 1 if there exists $i$ such that $\overline{k}_i>1$. 
\end{proposition}
\begin{proof}
We use subscripts 1 and 2 to represent the variance with respect to the probability distribution for models 1 and 2, respectively. We substitute Eq.~\eqref{m2q^*} in Eq.~\eqref{si1} and use Eq.~\eqref{Di2} to obtain
\begin{align}
\sigma^2_2[k_i] - \sigma^2_1[k_i] & = \frac{\overline{k}_i{a}^2b}{(a+b)^3}\left(1+\frac{\overline{k}_ia}{a+b}\right)-\overline{k}_iq^*(1-q^*)\nonumber\\
& = \frac{\overline{k}_i{a}^2b}{(a+b)^3}\left(1+\frac{\overline{k}_ia}{a+b}\right) -\overline{k}_i\left(\frac{a}{a+b}\right)^2\left[1-\left(\frac{a}{a+b}\right)^2\right]\nonumber\\
& = \overline{k}_i(\overline{k}_i-1)\frac{{a}^3b}{\left(a+b\right)^4}\nonumber\\
& > 0.
\end{align}

Next, we compare the variance of the average degree. Because $M$ is the number of edges of static network $G$ and there exists $i$ such that $\overline{k}_i>1$ in $G$, we obtain
\begin{equation}
\label{combination}
\sum_{i=1}^{N-1}\sum_{j=i+1}^N\sum\limits_{\substack{k>i\\k\neq j}}^{N-1}\sum\limits_{\substack{\ell>k\\\ell\neq i,j}}^N A_{ij}A_{k\ell} < \frac{M(M-1)}{2}
\end{equation}
for the following reason. The right-hand side of Eq.~\eqref{combination} is the number of pairs of edges. The left-hand side is the number of pairs of edges that do not share a node. These two quantities are equal to each other if and only if there is no pair of edges sharing a node, i.e., when all nodes have the degree at most 1.

By substituting Eq.~\eqref{m2q^*} in Eq.~\eqref{s1} and using Eqs.~\eqref{D2} and \eqref{combination}, we obtain
\begin{align}
& \sigma^2_2\left[\langle k\rangle\right]-\sigma^2_1\left[\langle k\rangle\right] \nonumber\\ 
= &\frac{4M}{N^2}\left(\frac{a}{a+b}\right)^2 + \frac{4M(M-1)}{N^2}\left(\frac{a}{a+b}\right)^3 - \frac{8}{N^2}\left(\frac{a}{a+b}\right)^3\frac{b}{a+b}\sum_{i=1}^{N-1}\sum_{j=i+1}^N\sum\limits_{\substack{k>i\\k\neq j}}^{N-1}\sum\limits_{\substack{\ell>k\\\ell\neq i,j}}^N A_{ij}A_{k\ell}\nonumber\\ 
& \quad - \left(\frac{2M}{N}\right)^2 \left(\frac{a}{a+b}\right)^4 - \frac{4Mq^*(1-q^*)}{N^2}\nonumber\\
> &\frac{4M}{N^2}\left(\frac{a}{a+b}\right)^2 + \frac{4M(M-1)}{N^2}\left(\frac{a}{a+b}\right)^3 - \frac{4M(M-1)}{N^2}\left(\frac{a}{a+b}\right)^3\frac{b}{a+b} \nonumber\\ 
& \quad - \left(\frac{2M}{N}\right)^2 \left(\frac{a}{a+b}\right)^4 - \frac{4M}{N^2}\left(\frac{a}{a+b}\right)^2\left[1 - \left(\frac{a}{a+b}\right)^2\right]\nonumber\\
= &0.
\end{align}
\end{proof}

\begin{proposition} \label{5}
Assume that $\mathbb E[k_i]$ is the same between models 1 and 3 for any $i\in \{1, \ldots, N\}$.
For any given $q^*$, the variance $\sigma^2[k_i]$ is larger for model 3 than model 1 if $\overline{k}_i>1$. Likewise, $\sigma^2[\langle k\rangle]$ is larger for model 3 than model 1 if there exists $i$ such that $\overline{k}_i>1$.
\end{proposition}
\begin{proof}
We substitute Eq.~\eqref{m3q^*} in Eq.~\eqref{si1} and use Eq.~\eqref{Di3} to obtain
\begin{align}
& \sigma^2_3[k_i] - \sigma^2_1[k_i]\nonumber\\ 
= & \overline{k}_i\frac{a(a+2b)}{\left(a+b\right)^2}+\overline{k}_i(\overline{k}_i-1)\frac{a({a}^2+3ab+{b}^2)}{\left(a+b\right)^3}-\left[\overline{k}_i\frac{a(a+2b)}{\left(a+b\right)^2}\right]^2-\overline{k}_iq^*(1-q^*)\nonumber\\
= & \overline{k}_i(\overline{k}_i-1)\frac{a{b}^3}{\left(a+b\right)^4}\nonumber\\
> & 0.
\end{align}
By substituting Eq.~\eqref{m3q^*} in Eq.~\eqref{s1} and using Eqs.~\eqref{D3} and \eqref{combination}, we obtain
\begin{align}
& \sigma^2_3\left[\langle k\rangle\right] - \sigma^2_1\left[\langle k\rangle\right] \nonumber\\
= & \frac{4M}{N^2}\cdot\frac{a(a+2b)}{(a+b)^2}+\frac{4M(M-1)}{N^2}\cdot\frac{a({a}^2+3a b+{b}^2)}{(a+b)^3}\nonumber\\
& \quad \ - \frac{8}{N^2}\cdot\frac{ab^3}{(a+b)^4}\sum_{i=1}^{N-1}\sum_{j=i+1}^N\sum\limits_{\substack{k>i\\k\neq j}}^{N-1}\sum\limits_{\substack{\ell>k\\\ell\neq i,j}}^N A_{ij}A_{k\ell} - \left(\frac{2M}{N}\right)^2 \frac{{a}^2(a+2b)^2}{(a+b)^4} - \frac{4Mq^*(1-q^*)}{N^2}\nonumber\\
> & \frac{4M}{N^2}\cdot\frac{a(a+2b)}{\left(a+b\right)^2}+\frac{4M(M-1)}{N^2}\cdot\frac{a({a}^2+3ab+{b}^2)}{\left(a+b\right)^3}\nonumber\\
& \quad \ - \frac{4M(M-1)}{N^2}\cdot\frac{a{b}^3}{\left(a+b\right)^4} - \left(\frac{2M}{N}\right)^2 \frac{{a}^2(a+2b)^2}{(a+b)^4} - \frac{4M}{N^2}\frac{ab^2(a+2b)}{(a+b)^4}\nonumber\\
= & 0.
\end{align}
\end{proof}

\begin{proposition} \label{6}
Assume that $\mathbb E[k_i]$ is the same between models 2 and 3 for any $i\in \{1, \ldots, N\}$.
For any given $q^*$, if $\overline{k}_i>1$, we obtain
\begin{enumerate}[(i)]
\item $\sigma^2_3[k_i] > \sigma^2_2[k_i]$ if $0<q^*<\frac{1}{2}$,
\item $\sigma^2_3[k_i] <\sigma^2_2[k_i]$ if $\frac{1}{2}<q^*\leq 1$,
\item $\sigma^2_3[k_i] = \sigma^2_2[k_i]$ if $q^*=\frac{1}{2}$.
\end{enumerate}
Furthermore, $\sigma^2[\langle k\rangle]$ satisfies the same relationships if there exists $i$ such that $\overline{k}_i>1$.
\end{proposition}
We prove Proposition~\ref{6} in Appendix~\ref{Appendix: C}.

\subsection{Duration for the edge being inactive in model 2}

In model 2, an edge is active if and only if both nodes forming the edge are in the $h$ state, and the state of each node (i.e., $h$ or $\ell$) independently obeys a continuous-time Markov process with two states. Therefore, the duration of the edge being active obeys an exponential distribution with rate $2b$. In contrast, the duration of the edge being inactive does not obey an exponential distribution, which we characterize as follows. 
\begin{proposition} \label{non-exponential}
The probability density function of the duration of the edge being inactive in model 2 is the mixture of two exponential distributions given by
\begin{equation} \label{mixdistribution}
f(t) = C_1\lambda_1e^{-\lambda_1t}+C_2\lambda_2e^{-\lambda_2t},
\end{equation}
where
\begin{align}
\lambda_1 =& \frac{1}{2}\left(3a+b-\sqrt{a^2+6ab+b^2}\right)\label{lambda1},\\
\lambda_2 =& \frac{1}{2}\left(3a+b+\sqrt{a^2+6ab+b^2}\right)\label{lambda2},\\
C_1 =& \frac{a\left(1+\frac{a-b}{\sqrt{a^2+6ab+b^2}}\right)}{3a+b-\sqrt{a^2+6ab+b^2}}\label{C1},
\end{align}
and
\begin{align}
C_2 =& \frac{a\left(1-\frac{a-b}{\sqrt{a^2+6ab+b^2}}\right)}{3a+b+\sqrt{a^2+6ab+b^2}}\label{C2}.
\end{align}
\end{proposition}
Note that it is straightforward to check $C_1+C_2=1$, $C_1 > 0$, $C_2 > 0$, and
$\lambda_1 > 0$.
\begin{proof}
Consider a three-state continuous-time Markov process described as follows. Let $z$ (with $z$= 0, 1, or 2) denote the number of nodes forming an edge that are in the $h$ state. We refer to the value of $z$ as the state of the three-state Markov chain without arising confusion with the single node's state (i.e., $h$ or $\ell$) or edge's state (i.e., active or inactive).
We initialize the stochastic dynamics of nodes by setting $z=2$, which corresponds to the edge being active. Consider a sequence of the state $z$ that starts from $z=2$ at time 0 and return to $z=2$ for the first time. Let $I_n$ be such a sequence of the $z$ values visiting $z=0$ in total $n$ times before returning to $z=2$ for the first time, which we denote by
\begin{equation*}
I_n = (2, 1,\underbrace{0, 1, ..., 0, 1}_{\text{$n$ repetitions of 0 and 1}},2). 
\end{equation*}
The duration of the edge being inactive is the difference between the time of the first passage to $z=2$ and the time of leaving $z=2$ last time.
We denote this duration by $T_n$ for the case in which $z=0$ is visited $n$ times before $z=2$ is revisited for the first time.

For general $n$, we obtain
\begin{equation}
T_n = \tau^{\prime}_{1,1} + \tau^{\prime\prime}_{0,1} + \tau^{\prime}_{1,2} + \cdots + \tau^{\prime\prime}_{0,n} + \tau^{\prime}_{1,n+1}, 
\end{equation}
where $\tau^{\prime}_{1,i}$ is the $i$th duration of $z=1$, which obeys the exponential distribution with rate $a+b$; $\tau^{\prime\prime}_{0,i}$ is the $i$th duration of $z=0$, which obeys the exponential distribution with rate $2a$. Variables $\tau^{\prime}_{1,1}$, $\ldots$, $\tau^{\prime}_{1,n+1}$, $\tau^{\prime\prime}_{0,1}$, $\ldots$, $\tau^{\prime\prime}_{0,n}$ are independent of each other. Therefore, the Laplace transform of the distribution of $T_n$ is given by
\begin{align}
\mathcal{L}_n(s) & = \mathbb E[e^{-sT_n}] \nonumber\\
& = \mathbb E[e^{-s\tau^{\prime}_{1,1}}]\mathbb E[e^{-s\tau^{\prime\prime}_{0,1}}]\mathbb E[e^{-s\tau^{\prime}_{1,2}}]\mathbb E[e^{-s\tau^{\prime\prime}_{0,2}}] \cdots \mathbb E[e^{-s\tau^{\prime\prime}_{0,n}}]\mathbb E[e^{-s\tau^{\prime}_{1,n+1}}] \nonumber\\
& = \left(\frac{a+b}{s+a+b}\right)^{n+1}\left(\frac{2a}{s+2a}\right)^n.
\label{eq:L_n(s)}
\end{align}
The probability that sequence $I_n$ occurs is given by 
\begin{equation}
\tilde{q}(n)=1\cdot \left(\frac{b}{a+b}\right)^n \cdot 1^n \cdot \frac{a}{a+b} = \frac{ab^n}{(a+b)^{n+1}}.
\label{eq:q(n)}
\end{equation}
Let $T$ be the duration for which the edge is inactive. Using Eqs.~\eqref{eq:L_n(s)} and \eqref{eq:q(n)}, we obtain the Laplace transform of the distribution of $T$ as follows:
\begin{align}
\mathcal{L}(s) & = \sum_{n=0}^{\infty} \tilde{q}(n)\mathbb E[e^{-sT_n}] \nonumber\\
& = \sum_{n=0}^{\infty} \frac{ab^n}{(a+b)^{n+1}}\left(\frac{a+b}{s+a+b}\right)^{n+1}\left(\frac{2a}{s+2a}\right)^n\nonumber\\
& = \frac{a(s+2a)}{s^2+(3a+b)s+2a^2}\nonumber\\
& = \frac{a\left(s+\frac{3a+b}{2}+\frac{a-b}{2}\right)}{\left(s+\frac{3a+b}{2}\right)^2-\left(\sqrt{\frac{a^2+6ab+b^2}{4}}\right)^2}.
\end{align}
Therefore,
\begin{align}
\mathcal{L}^{-1}(s) & = a\left[e^{-\frac{3a+b}{2}t}\cosh{\left(\sqrt{\frac{a^2+6ab+b^2}{4}}t\right)}+\frac{a-b}{\sqrt{a^2+6ab+b^2}}e^{-\frac{3a+b}{2}t}\sinh{\left(\sqrt{\frac{a^2+6ab+b^2}{4}}t\right)}\right]\nonumber\\
& = C_1\lambda_1e^{-\lambda_1t}+C_2\lambda_2e^{-\lambda_2t}.
\end{align}
\end{proof}

We verify Eq.~\eqref{mixdistribution} with numerical simulations for two parameter sets. The results are shown in Fig.~\ref{fig:test}. The dashed curves represent the exponential distributions whose mean is the same as that of the corresponding mixture of two exponential distributions. The figure suggests that the actual duration for the edge to be inactive is distributed more heterogeneously than the exponential distribution for both parameter sets.

\floatsetup[figure]{style=plain,subcapbesideposition=top}
\captionsetup{font={small,rm}} 
\begin{figure}[H]
  \centering
  \includegraphics[scale=0.45]{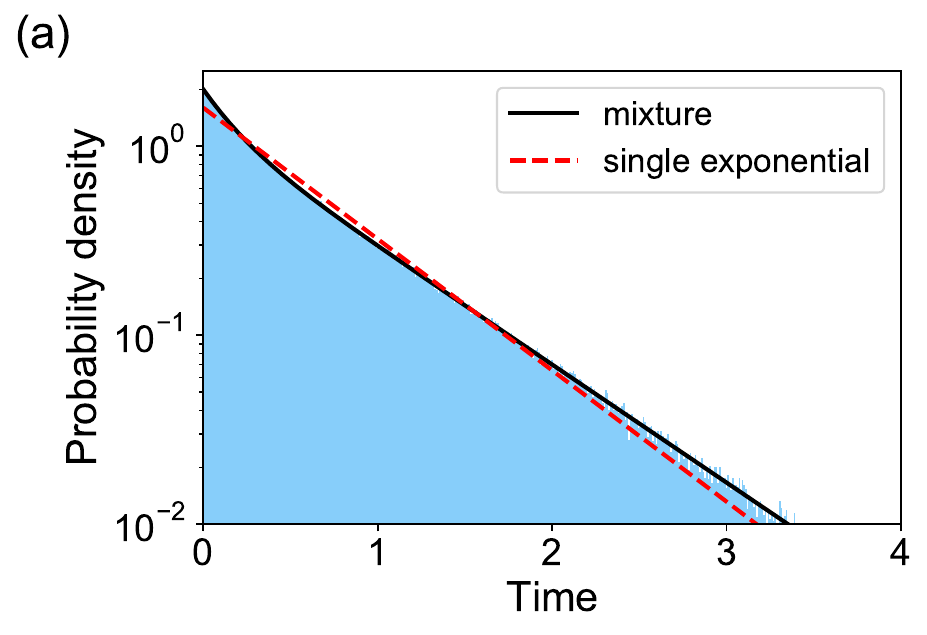}
  \hspace{0.3in}
  \includegraphics[scale=0.45]{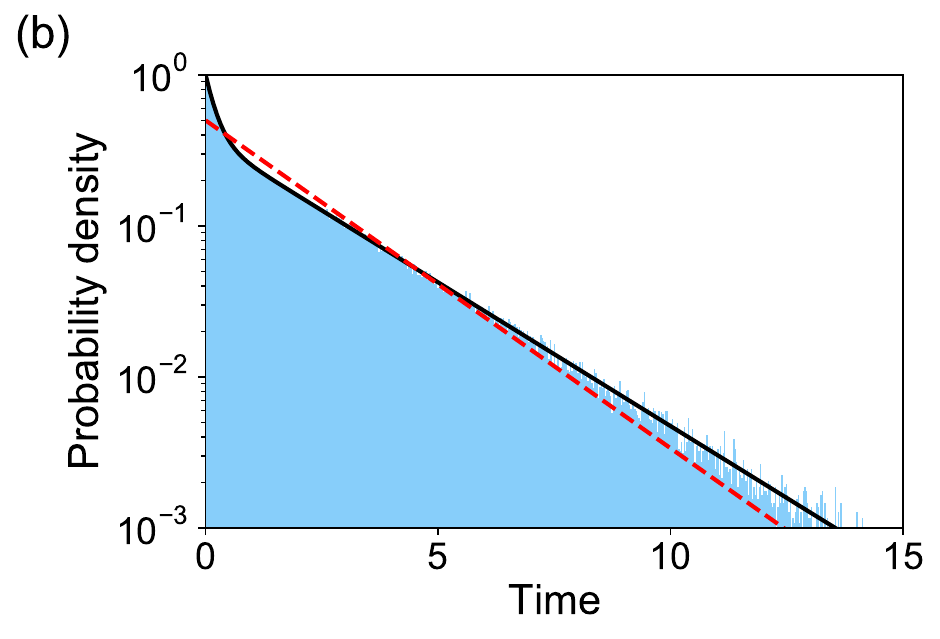}
  \caption{Distributions of the duration of the edge being inactive in model 2.
  %
  %
  The shaded bars represent numerically obtained distributions calculated on the basis of $5\times 10^5$ samples. The solid lines represent the mixture of two exponential distributions, i.e., Eq.~\eqref{mixdistribution}. The dashed lines represent the exponential distribution whose mean is the same as that for Eq.~\eqref{mixdistribution}. \textbf{(a)} $a=2.0$, $b=1.0$. \textbf{(b)} $a=1.0$, $b=2.0$.}
  \label{fig:test}
\end{figure}

For model 3, the duration for the edge being inactive, which is equivalent to the time for which the same three-state Markov process stays in state $z=0$,
obeys an exponential distribution with rate $2a$. The duration for the edge being active is the first passage time to $z=0$ since the Markov chain has left $z=0$. Therefore, the duration for the edge being active for model 3 obeys the mixture of two exponential distributions given by Eq.~\eqref{mixdistribution}, but with $a$ and $b$ being swapped.

\section{Impact of concurrency on the dynamics of the SIS and SIR models}

\subsection{Numerical results}

\subsubsection{Methods} 

To examine the effect of concurrency on epidemic spreading, we run the stochastic SIS and SIR models on three static networks. On each static network, we run stochastic dynamics of partnership according to model 1, 2, or 3 and compare the extent of epidemic spreading among the models. Each node takes either the susceptible, infectious, or recovered state, and the node's state may change over time. In both SIS and SIR models, an infectious node infects each of its susceptible neighbor independently at rate $\beta$, which we call the infection rate, and an infectious node recovers at rate $\mu$, which we call the recovery rate, independently of the other nodes' states. The infection and recovery events occur as Poisson processes with the respective rates. Once an infectious node recovers, it turns into the susceptible state in the case of the SIS model and the recovered state in the case of the SIR model. Recovered nodes in the SIR model do not infect other nodes and are not infected by other nodes.

We consider three static networks. First, we use the Erd\H{o}s-R\'{e}nyi (ER) random graph with $N=200$ nodes. We independently connect each pair of nodes with probability $0.05$. We iterated generating a network from the ER random graph until we obtained a connected network with $M=1000$ edges. 
The second network is a network with $N=200$ nodes that has a degree distribution with a power-law tail that is proportional to $k^{-3}$, where $k$ is the degree, generated by
the Barab\'{a}si-Albert (BA) model \cite{Barabasi1999Sci}.
We assume that each incoming node has five edges to be connected to already existing nodes according to the proportional preferential attachment rule. In other words, the probability that an existing node, denoted by $v$, forms a new edge with an incoming node is proportional to $v$'s degree. The initial network is a star graph on 6 nodes. The network is connected and contains $M=975$ edges. 
Third, we use the largest connected component of
a collaboration network among researchers who had published papers in network science up to 2006 \cite{Newman3}. The network has $N=379$ nodes and $M=914$ edges. The edge represents the presence of at least one paper that two authors have coauthored. We use this network as an unweighted network. 

We set $\mu=1$ without loss of generality; multiplying the same constant to $\beta$, $\mu$, $a_1$, $b_1$, $a_2$, $b_2$, $a_3$, and $b_3$ only rescales the time; $a_i$ and $b_i$ represent parameters $a$ and $b$ for model $i$. We run 5000 simulations for each epidemic process model (i.e., SIS or SIR), each static network, each partnership model (i.e., model 1, 2, or 3), and each parameter set.
In each simulation of the SIS model, we started from the initial condition in which all nodes are infectious. This initial condition is not realistic. However, we use it with the aim of identifying the quasi-stationary fraction of infectious nodes without being affected by rapid extinction of infectious nodes that may happen with an initial condition having only a small number of infectious nodes even if the infection rate is large. In the SIS model, we measure the fraction of infectious nodes at
time $t = 10^4$ and average it over $10^3$ simulations at each $\beta$ value.
In each simulation of the SIR model, just one node selected uniformly at random is initially infectious, and all the other nodes are susceptible. For each static network,
the initially infectious node is the same over the different partnership models and the different $\beta$ values. For each $\beta$ value, we 
average the fraction of recovered nodes at the end of the simulation over the 5000 simulations, which we call the final epidemic size. 

We run the simulation with model 1 using the Laplace Gillespie algorithm \cite{Masuda2018SIAM, Masuda2021arxiv}, which is an extension of the direct method of Gillespie to the case in which inter-event times are allowed to obey a mixture of exponential distributions like Eq.~\eqref{mixdistribution}. For models 2 and 3, we run the simulation using the direct method of Gillespie because these models only involve exponential distributions of inter-event times.

Apart from the variation in the $\beta$ value, we consider four sets of parameter values for both SIS and SIR models. In the first set of simulations, we set $a_1=1$, $b_1=9$, $a_2=0.5(\sqrt{10}+1)\approx 2.08$, $b_2=4.5$, $a_3=0.5$, and $b_3=1.5(\sqrt{10}+3)\approx 9.24$. Then, we obtain $q^*=0.1$ for all the three partnership models. In this manner, we compare the final epidemic size for the three partnership models, which yield different amounts of concurrency, under the condition that each edge is active for the same amount of time on average. It should be noted that a large $q^*$ value will obviously lead to a larger final epidemic size with the other things being equal, and so it is necessary to compare the models with the $q^*$ value being fixed. We obtain a second parameter set by making the values of $a_1$, $b_1$, $a_2$, $b_2$, $a_3$, and $b_3$ for the first parameter set five times smaller, which implies that the nodes and edges flip their states five times more slowly than in the case of the first parameter set. Because multiplying these six parameters by the same constant does not change $q^*$, we retain $q^* = 0.1$ for the second parameter set. We refer to the first and second parameter sets as the fast and slow edge dynamics, respectively.
The third parameter set is defined by $a_1=1$, $b_1=1$, $a_2=0.5(\sqrt{2}+1)\approx 1.21$, $b_2=0.5$, $a_3=0.5$, and $b_3=0.5(\sqrt{2}+1)\approx 1.21$ such that $q^*=0.5$. We also consider a slower variant of edge dynamics by dividing these $a_1$, $b_1$, $a_2$, $b_2$, $a_3$, and $b_3$ values by five. In summary, for each of the three networks and each partnership model, we have four cases each of which consists of the combination of
$q^* \in \{0.1, 0.5\}$ and either the fast or slow edge dynamics.

\subsubsection{Results for the SIS model\label{sub:SIS-numerical}}

In this section, we examine the SIS model with partnership models $1$, $2$, and $3$.
We first compare between partnership models $1$ and $2$. We recall that model 2 has a higher concurrency than model 1 when each edge is active with the same probability in the two models. To exclude the possible effects of the distribution of the duration for which the edge is active and that for which the edge is inactive,
we use model $1$ with $\psi_1(\tau_1)=2b_2e^{-2b_2\tau_1}$ and $\psi_2(\tau_2)=C_1\lambda_1e^{-\lambda_1\tau_2}+C_2\lambda_2e^{-\lambda_2\tau_2}$, where $\lambda_1$, $\lambda_2$, $C_1$, and $C_2$ are given by Eqs.~\eqref{lambda1}, \eqref{lambda2}, \eqref{C1}, and \eqref{C2}, respectively, with $a = a_2$ and $b = b_2$. In this manner, models 1 and 2 have the identical distribution of the duration of the edge being active (i.e., $\psi_1$) and that of the edge being inactive (i.e., $\psi_2$), whereas they are different in terms of the amount of concurrency.

In Fig.~\ref{fig:sis_1p_vs_2}(a)--(d), we show the relationships between the infection rate and fraction of infectious nodes at $t=10^4$ for a network generated by the ER random graph. Figure~\ref{fig:sis_1p_vs_2}(a) and (b) corresponds to $q^*=0.1$; Fig.~\ref{fig:sis_1p_vs_2}(c) and (d) corresponds to $q^*=0.5$. 
Figure~\ref{fig:sis_1p_vs_2}(a) and (c) corresponds to the slow edge dynamics; Fig.~\ref{fig:sis_1p_vs_2}(b) and (d) corresponds to the fast edge dynamics.
Figure~\ref{fig:sis_1p_vs_2}(a)--(d) indicates that, in each of these combinations of a value of $q^*$ and speed of the edge dynamics, the epidemic threshold is smaller for model $2$ than model $1$. As a result, the fraction of infectious nodes is larger for model $2$ than model $1$ when the infection rate is near the epidemic threshold. These results suggest that concurrency decreases the epidemic threshold and promotes epidemic spreading when the infection rate is near the epidemic threshold, which is consistent with the main conclusions of previous studies \cite{Bauch2000PRSB,Onaga2017PRL,Onaga2019book}.
However, the epidemic threshold and the fraction of infectious nodes near the epidemic threshold are almost the same between models 1 and 2 when $q^* = 0.5$ (see Fig.~\ref{fig:sis_1p_vs_2}(c) and (d)).
Furthermore, the concurrency decreases the fraction of infectious nodes when the infection rate is sufficiently larger than the epidemic threshold, with both $q^*=0.1$ and $q^*=0.5$.

The results for a network generated by the BA model, shown in Fig.~\ref{fig:sis_1p_vs_2}(e)--(h), and
those for the collaboration network, shown in Fig.~\ref{fig:sis_1p_vs_2}(i)--(l), are qualitatively the same as those for the ER model.
Quantitatively, when $q^* = 0.1$, the relative difference in the epidemic threshold between models 1 and 2 for the BA model (see Fig.~\ref{fig:sis_1p_vs_2}(e) and (f)) and the collaboration network (see Fig.~\ref{fig:sis_1p_vs_2}(i) and (j)) is smaller than for the ER random graph (see Fig.~\ref{fig:sis_1p_vs_2}(a) and (b)).

\floatsetup[figure]{style=plain,subcapbesideposition=top}
\captionsetup{font={small,rm}} 
\captionsetup{labelfont=bf}
\begin{figure}[H]
  \centering
  \includegraphics[width=1.0\linewidth]{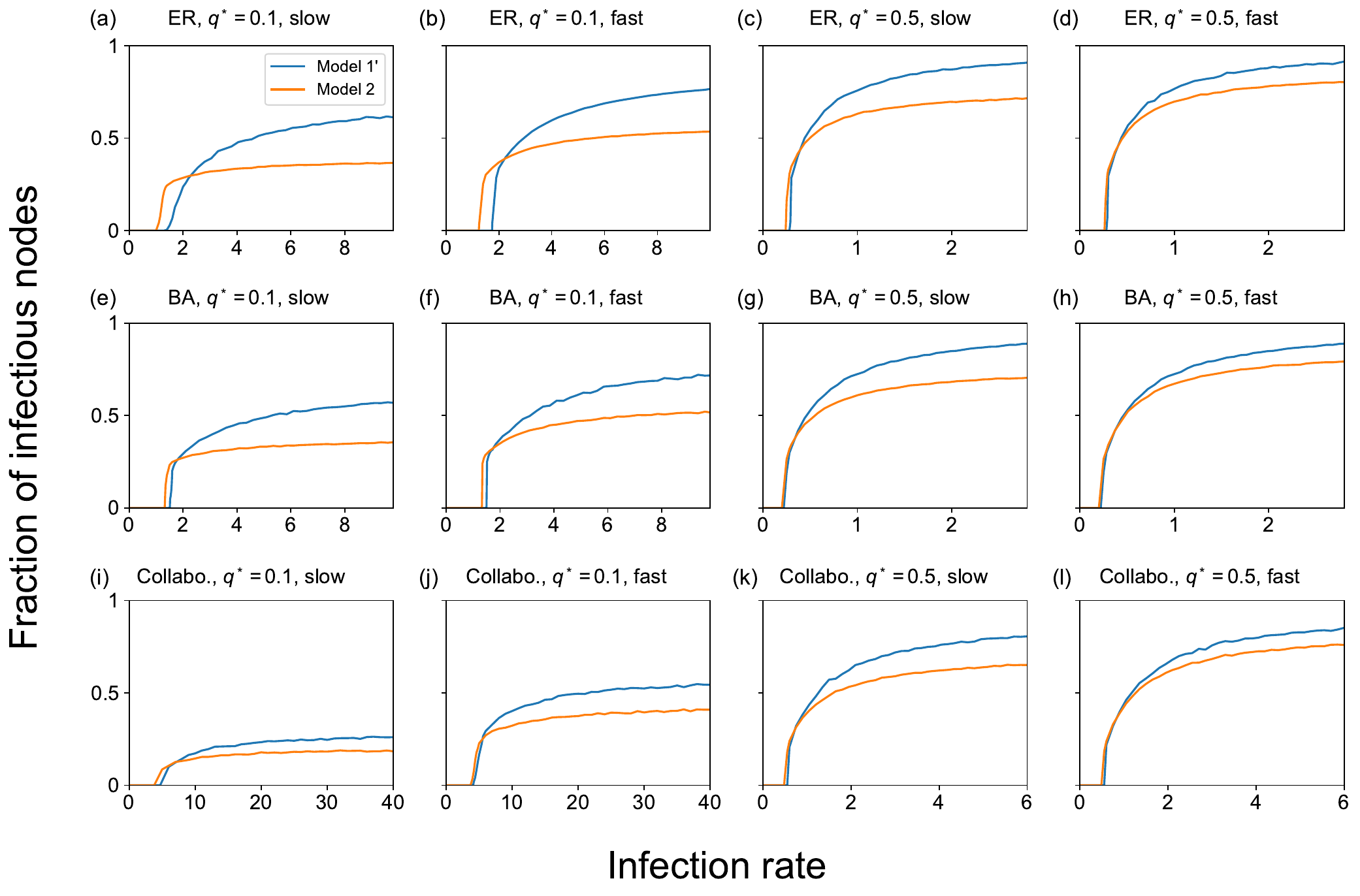}
   \caption{Comparison of the quasi-equilibrium fraction of infectious nodes in the SIS model between models $1$ and $2$. (a)--(d) ER random graph. (e)--(h) BA model. (i)--(l) Collaboration network. In panels (a), (e), and (i), we set $q^*=0.1$ and use the slow edge dynamics. In panels (b), (f), and (j), we set $q^*=0.1$ and use the fast edge dynamics. In panels (c), (g), and (k), we set $q^*=0.5$ and use the slow edge dynamics. In panels (d), (h), and (l), we set $q^*=0.5$ and use the fast edge dynamics. ``Collabo.''~is a shorthand for the collaboration (i.e., co-authorship) network.}   
   \label{fig:sis_1p_vs_2}
\end{figure}

Now we compare models $1$ and $3$. We use model $1$ with $\psi_1(\tau_1)=C_1\lambda_1e^{-\lambda_1\tau_1}+C_2\lambda_2e^{-\lambda_2\tau_1}$, where the four parameters are given by
Eqs.~\eqref{lambda1}, \eqref{lambda2}, \eqref{C1}, and \eqref{C2}, respectively, with $a = b_3$ and $b = a_3$, and $\psi_2(\tau_2)=2a_3e^{-2a_3\tau_2}$. In this manner, models 1 and 3 have the identical $\psi_1$ and $\psi_2$, whereas model 3 has a larger amount of concurrency than model 1 does.

We show the fraction of infectious nodes at $t=10^4$ for models $1$ and $3$ for the same three underlying static networks, two values of $q^*$, and two speeds of edge dynamics in Fig.~\ref{fig:sis_1dp_vs_3}. The results are qualitatively similar to the comparison between models $1$ and $2$ shown in Fig.~\ref{fig:sis_1p_vs_2}. These results with models 1 and 3 further support that concurrency decreases the epidemic threshold, whereas this effect is small for the BA model and collaboration network as well as for denser networks (i.e., $q^* = 0.5$ compared to $q^* = 0.1$).

\floatsetup[figure]{style=plain,subcapbesideposition=top}
\captionsetup{font={small,rm}} 
\captionsetup{labelfont=bf}
\begin{figure}[H]
  \centering
  \includegraphics[width=1.0\linewidth]{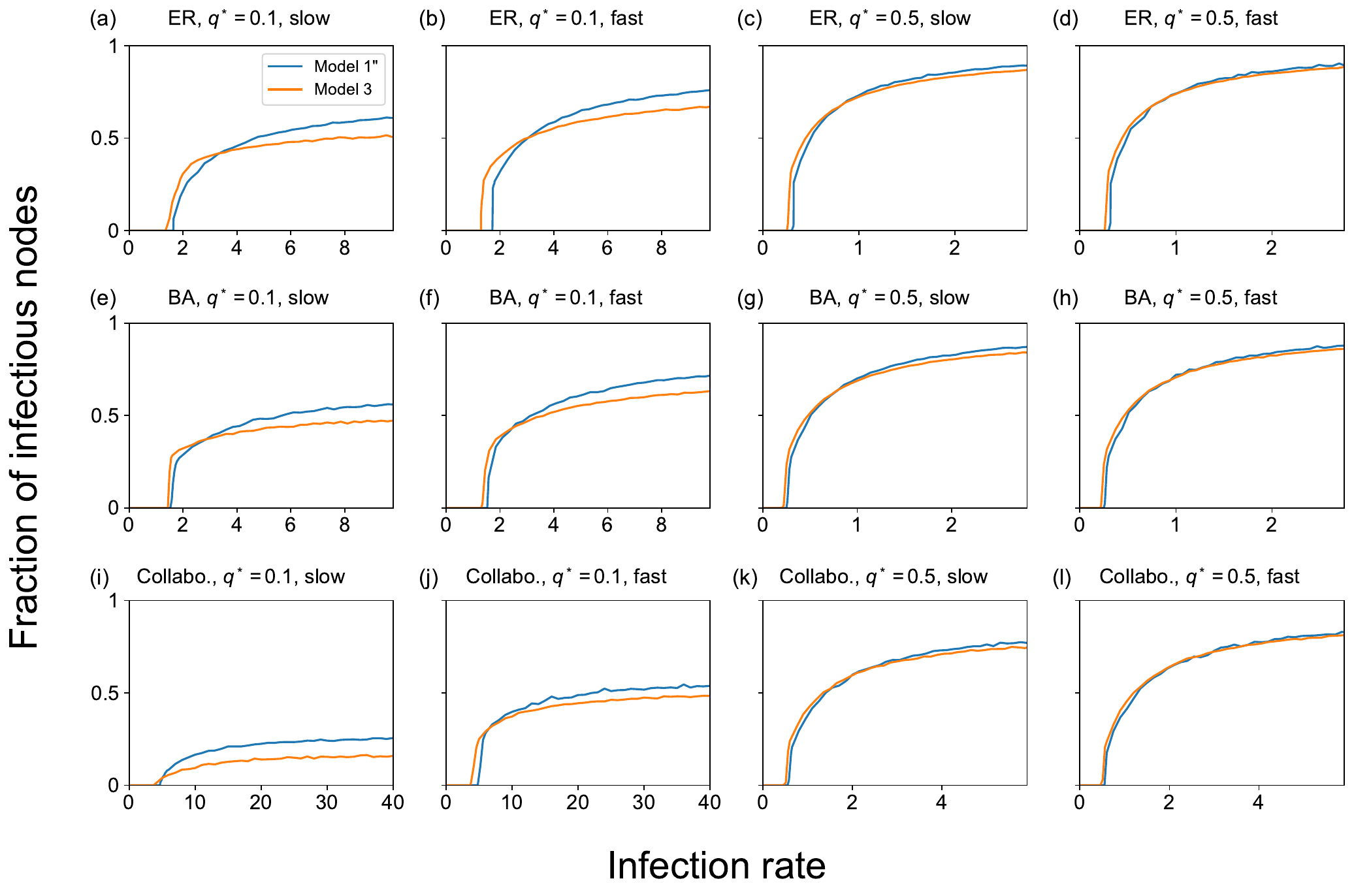}
   \caption{Comparison of the quasi-equilibrium fraction of infectious nodes in the SIS model between models $1$ and $3$. (a)--(d) ER random graph. (e)--(h) BA model. (i)--(l) Collaboration network. In panels (a), (e), and (i), we set $q^*=0.1$ and use the slow edge dynamics. In panels (b), (f), and (j), we set $q^*=0.1$ and use the fast edge dynamics. In panels (c), (g), and (k), we set $q^*=0.5$ and use the slow edge dynamics. In panels (d), (h), and (l), we set $q^*=0.5$ and use the fast edge dynamics.}   
   \label{fig:sis_1dp_vs_3}
\end{figure}

\subsubsection{Results for the SIR model\label{sub:SIR-numerical}}

In this section, we examine the SIR model with the three partnership models. When comparing between models 1 and 2 and between models 1 and 3, we use the same distributions of inter-event times for model 1 as those used in section~\ref{sub:SIS-numerical}.

In Fig.~\ref{fig:1p_vs_2}(a)--(d), we show the relationships between the infection rate and final epidemic size for a network generated by the ER random graph. Figure~\ref{fig:1p_vs_2}(a) and (b) corresponds to $q^*=0.1$; Fig.~\ref{fig:1p_vs_2}(c) and (d) corresponds to $q^*=0.5$. Figure~\ref{fig:1p_vs_2}(a) and (c) corresponds to the slow edge dynamics; Fig.~\ref{fig:1p_vs_2}(b) and (d) corresponds to the fast edge dynamics.
Figure~\ref{fig:1p_vs_2}(a)--(d) indicates that, in each of these combinations of a value of $q^*$ and speed of the edge dynamics, the final epidemic size is larger for model $2$ than model $1$ when the infection rate is near the epidemic threshold. This result suggests that concurrency promotes epidemic spreading near the epidemic threshold. However, the final epidemic size near the epidemic threshold is only slightly larger for model 2 than model 1 when $q^*= 0.5$ (see Fig.~\ref{fig:1p_vs_2}(c) and (d)). Furthermore, the concurrency suppresses the epidemic spreading when the infection rate, and hence the final epidemic size, is larger for both $q^*=0.1$ and $q^*=0.5$.
All these results are similar to those for the SIS model shown in Fig.~\ref{fig:sis_1p_vs_2}(a)--(d) except that there is no noticeable difference in the epidemic threshold between models 1 and 2 in the case of the SIR model.

The results for the BA model, shown in Fig.~\ref{fig:1p_vs_2}(e)--(h), are similar to those for the ER random graph. The results for the collaboration network, shown in Fig.~\ref{fig:1p_vs_2}(i)--(l), are similar to those for the ER random graph and the BA model except that, when $q^*=0.1$, the final epidemic size is larger for model 2 than model 1 across the entire range of the infection rate, $\beta$, that we investigated; see
Fig.~\ref{fig:1p_vs_2}(i) and (j). Because the final epidemic size has approximately plateaued for both models 1 and 2 at the largest $\beta$ value that we investigated, it is unlikely that model 1 yields a larger final epidemic size than model 2 when $\beta$ is larger.
These numerical results suggest that the effect of concurrency on enhancing epidemic spreading at $q^* = 0.1$ is stronger for the collaboration network than for the ER and BA models.
In contrast, when $q^*=0.5$, the final epidemic size is not noticeably larger for model 2 than model 1 when $\beta$ is near the epidemic threshold and smaller for model 2 than model 1 when $\beta$ is larger (see Fig.~\ref{fig:1p_vs_2}(k) and (l)). This result is qualitatively the same as that for the ER and BA models shown in Fig.~\ref{fig:1p_vs_2}(c), (d), (g), and (h).

\floatsetup[figure]{style=plain,subcapbesideposition=top}
\captionsetup{font={small,rm}} 
\captionsetup{labelfont=bf}
\begin{figure}[H]
  \centering
  \includegraphics[width=1.0\linewidth]{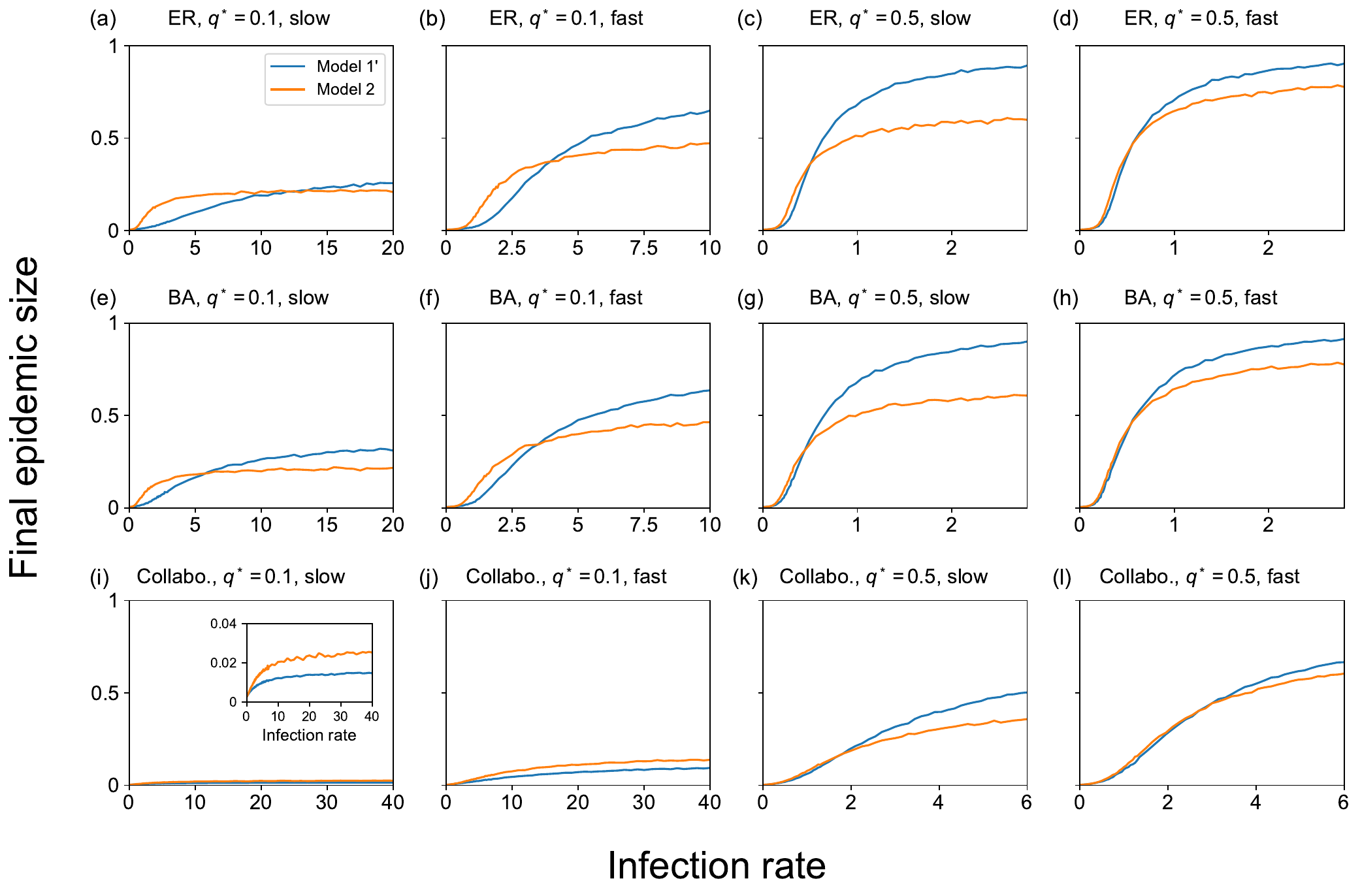}
   \caption{Comparison of the final epidemic size in the SIR model between models $1$ and $2$. (a)--(d) ER random graph. (e)--(h) BA model. (i)--(l) Collaboration network. In panels (a), (e), and (i), we set $q^*=0.1$ and use the slow edge dynamics. In panels (b), (f), and (j), we set $q^*=0.1$ and use the fast edge dynamics. In panels (c), (g), and (k), we set $q^*=0.5$ and use the slow edge dynamics. In panels (d), (h), and (l), we set $q^*=0.5$ and use the fast edge dynamics. The inset of (i) shows the magnification of the main panel because the final epidemic size is small in the entire range of the infection rate in this case.}   
   \label{fig:1p_vs_2}
\end{figure}

In Fig.~\ref{fig:1dp_vs_3}, we compare the final epidemic size in the SIR model between models $1$ and $3$ for the same three underlying static networks, two values of $q^*$, and two speeds of edge dynamics. The results are largely similar to the comparison between models $1$ and $2$. A notable difference from the comparison between models $1$ and $2$ is that, for the ER random graph and BA model combined with $q^*=0.1$ and the slow edge dynamics,
model 3 yields a larger final epidemic size than model 1 across the entire range of the infection rate values investigated (see Fig.~\ref{fig:1dp_vs_3}(a) and (e)).

\floatsetup[figure]{style=plain,subcapbesideposition=top}
\captionsetup{font={small,rm}} 
\captionsetup{labelfont=bf}
\begin{figure}[H]
  \centering
  \includegraphics[width=1.0\linewidth]{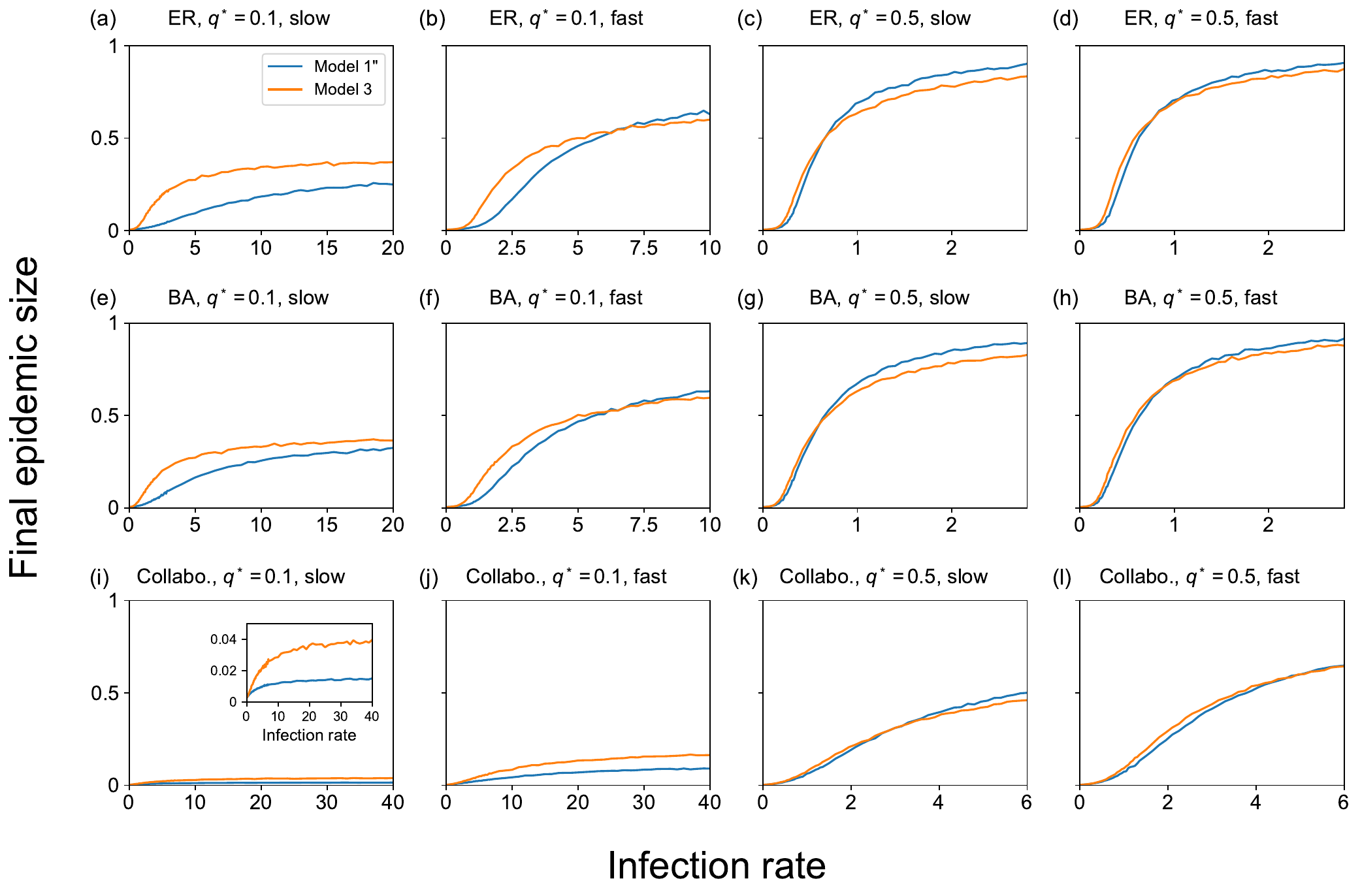}
   \caption{Comparison of the final epidemic size in the SIR model between models $1$ and $3$. (a)--(d) ER random graph. (e)--(h) BA model. (i)--(l) Collaboration network. In panels (a), (e), and (i), we set $q^*=0.1$ and use the slow edge dynamics. In panels (b), (f), and (j), we set $q^*=0.1$ and use the fast edge dynamics. In panels (c), (g), and (k), we set $q^*=0.5$ and use the slow edge dynamics. In panels (d), (h), and (l), we set $q^*=0.5$ and use the fast edge dynamics.}
   \label{fig:1dp_vs_3}
\end{figure}

\subsection{Theoretical results\label{sub:Ogura}}

We analytically derived an expression of the epidemic threshold for the SIS model on arbitrary underlying static networks combined with partnership models 1, 2, and 3. We assume that each network composed of partnership edges lasts for time $h$ before switching to another partnership network and that 
the static network $\mathcal{G}(t)$ in each time window of length $h$ is generated as an i.i.d. 
In fact, our partnership models 1, 2, and 3 imply that $\mathcal{G}(t)$ at different times are correlated with each other. Therefore, the i.i.d. assumption is for facilitating theoretical analysis.

We found that, to the first order of $h$, the epidemic threshold is the same for the three models. We show the detail in Appendix D.
This result is apparently not consistent with the numerical results for the SIS model shown in Figs.~\ref{fig:sis_1p_vs_2} and \ref{fig:sis_1dp_vs_3}, which support that the epidemic threshold decreases as the concurrency increases. In fact, the i.i.d. assumption made for our theoretical analysis corresponds to the limit of infinitely fast switching of edges. This is because a faster edge dynamics for a fixed value of $h$ implies that $\mathcal{G}(t)$ and $\mathcal{G}(t')$, where $t' \neq t$, are less correlated with each other for given $t$ and $t'$. Therefore, we carried out additional numerical simulations by making the values of $a_1$, $b_1$, $a_2$, $b_2$, $a_3$, and $b_3$ five times larger than those for the fast edge dynamics, i.e., 25 times larger than those for the slow edge dynamics. Then, we found that the epidemic threshold for the SIS model in the case of high concurrency (i.e., models 2 and 3) is much closer to the epidemic threshold in the case of low concurrency (i.e., model 1), as compared to when the edge dynamics are slower (i.e., Figs.~\ref{fig:sis_1p_vs_2} and \ref{fig:sis_1dp_vs_3}). See Fig.~\ref{fig:sis_even_faster_edge_dynamics} for numerical results. We conclude that the theoretical results in this section explain our numerical results when the edge dynamics are sufficiently faster than the epidemic dynamics.

\floatsetup[figure]{style=plain,subcapbesideposition=top}
\captionsetup{font={small,rm}} 
\captionsetup{labelfont=bf}
\begin{figure}[H]
  \centering
  \includegraphics[width=1.0\linewidth]{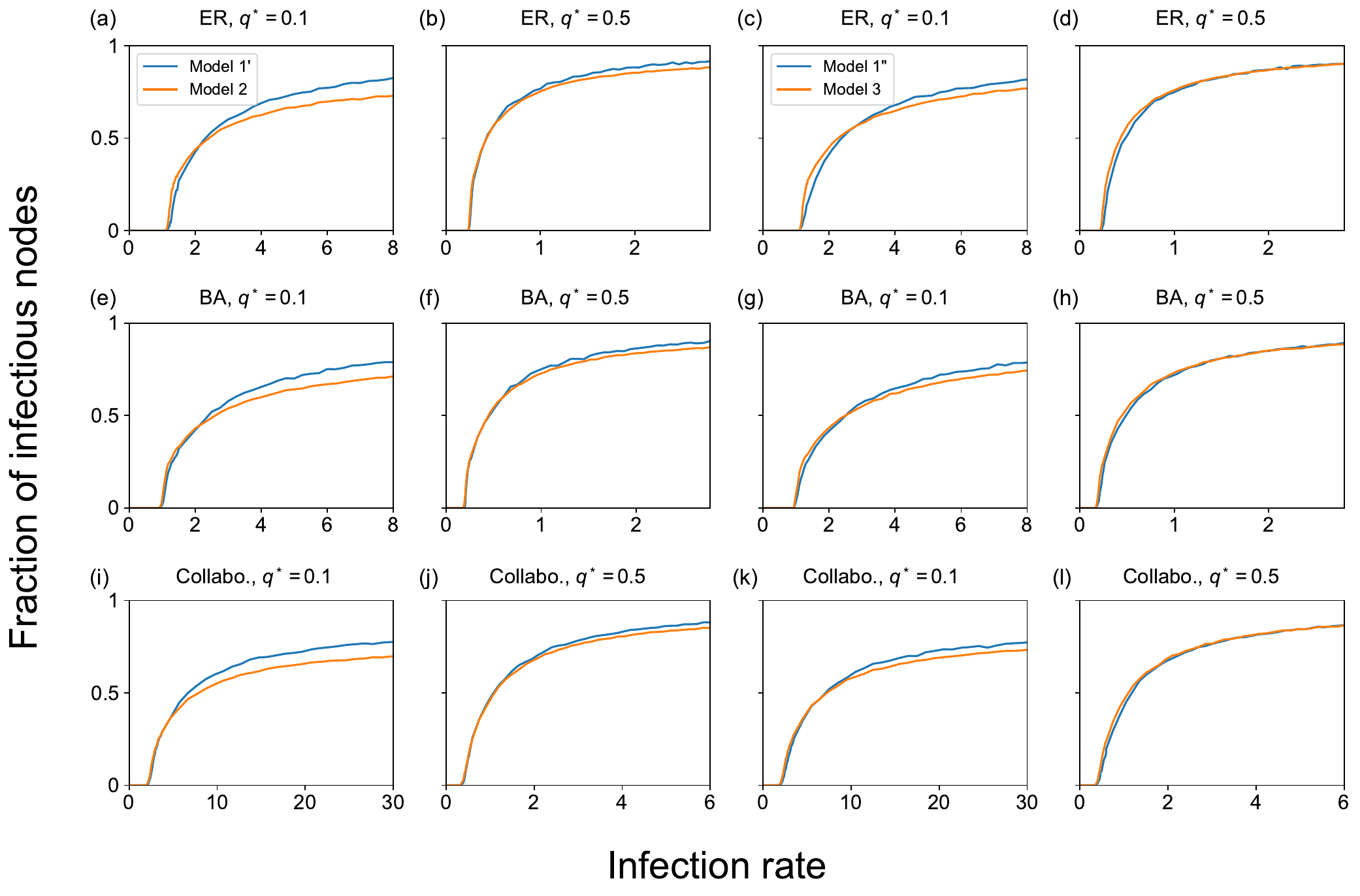}
   \caption{
Comparison of the quasi-stationary fraction of infectious nodes in the SIS model among models $1$, $2$, and $3$ under faster edge dynamics. We use the edge dynamics that are five times faster than the fast edge dynamics used in Figs.~\ref{fig:sis_1p_vs_2} and \ref{fig:sis_1dp_vs_3}.
(a)--(d) ER random graph. (e)--(h) BA model. (i)--(l) Collaboration network. In panels (a), (e), and (i), we set $q^*=0.1$ and compare models $1$ and $2$.
In panels (b), (f), and (j), we set $q^*=0.5$ and compare models $1$ and $2$. In panels (c), (g), and (k), we set $q^*=0.1$ and compare models $1$ and $3$.
In panels (d), (h), and (l), we set $q^*=0.5$ and compare models $1$ and $3$. We remind that we cannot compare models $1$, $2$, and $3$ in a single figure panel because we need to use different variants of model $1$ in the comparison with models $2$ and $3$ to make the comparison fair.}
   \label{fig:sis_even_faster_edge_dynamics}
\end{figure}

\section{Discussion}

We investigated three models of edge dynamics. They allow us to generate cases with different amounts of concurrency while keeping the probability that each edge is active (i.e., $q^*$) and the structure of the aggregate network the same. We theoretically evaluated the concurrency of each model and compared it among the models. Then, we numerically observed using the stochastic SIS and SIR models that quasi-equilibrium fractions of infectious nodes and the final epidemic size, respectively, were larger for the edge dynamics models with the higher concurrency (i.e., models 2 or 3 compared to model 1) if the infection rate was near the epidemic threshold. 
In the SIS model, the epidemic threshold was also smaller when the concurrency was higher.
These effects of the concurrency surrounding the epidemic threshold were larger when the edge dynamics is slower, i.e., with $q^*=0.1$ than with $q^*=0.5$.
At the same time, we found that the effect of concurrency on epidemic spreading was modest up to our numerical efforts. It was particularly the case for the BA model and the collaboration network, which are more realistic than the ER random graph in terms of the heterogeneity in the node's degree.
Furthermore, high concurrency suppressed epidemic spreading at large infection rates.

A previous study numerically showed that concurrency has a greater effect on increasing an epidemic potential in sparser networks \cite{Moody}. They measured the epidemic potential defined by the proportion of ordered pairs of nodes that are reachable in the sense that it is possible to travel from one node to the other node along a time-respecting path in the given temporal network composed of time-stamped edges. Our results are consistent with theirs because the effect of concurrency on enhancing epidemic spreading is stronger when $q^*$ is smaller in all the cases investigated (see Figs.~\ref{fig:sis_1p_vs_2}, \ref{fig:sis_1dp_vs_3}, \ref{fig:1p_vs_2}, and \ref{fig:1dp_vs_3}). The larger effect of concurrency on epidemic spreading for sparser networks may be because networks with large $q^*$, which leads to a large number of edges, have more time-respecting paths from an infectious node to susceptible nodes transmitting infection even in the absence of concurrency \cite{Moody}.

Another previous study concluded that the impact of concurrency on the epidemic size is fairly limited except in an early stage of epidemic dynamics \cite{Miller}. Despite the dependence of epidemic spreading on the density of contacts, which we discussed in the last paragraph, our findings are at large consistent with theirs. In other words, we found that the presence of concurrency little affected or even decreased the quasi-equilibrium fraction of infectious nodes in the SIS model and the final epidemic size in the SIR model in many cases, in particular for the BA model and collaboration networks. The reason why high concurrency decreases epidemic spreading at large infection rates is unclear. This phenomenon may be because, at high infection rates, infection can easily spread from nodes to nodes even without concurrency, such that many contact events on different edges, which would occur at close times in the case of high concurrency, are wasted. Note that our modeling framework guarantees that the total number of contact events on each edge is the same between low and high concurrency cases on average.
Our analytical result showing that the epidemic threshold is independent of the concurrency (i.e., see section~\ref{sub:Ogura}) also supports that the impact of concurrency on enhancing epidemic spreading may not be large.

In contrast, our previous study analytically showed that the concurrency enhances epidemic spreading in terms of the epidemic threshold value for the SIS model \cite{Onaga2017PRL, Onaga2019book}. The diversity of these results may be due to different strategies for modeling concurrency. For example, the authors of Ref.~\cite{Miller} designed their dynamic network models by ensuring that the number of partners that an individual has in a long term is the same across the individuals, rendering the aggregate network the complete graph. 
A different study that also carefully controlled the amount of interaction for each edge to be the same across comparisons made the same homogeneity assumption \cite{Bauch2000PRSB}. In contrast, Refs.~\cite{Onaga2017PRL, Onaga2019book} and the present study have assumed that the number of partners that an individual has in a long term may be heterogeneously distributed. Examining similarities and differences between these existing models of concurrency, including the model proposed in the present study, is an outstanding question. This being said, an overall conclusion based on the present numerical results is that the concurrency only modestly or little influences extents of epidemic spreading in many cases. This result indicates that other factors such as static network structure \cite{Newman2018book, Pastor2015RMP, Kiss2017book, Barrat2008book} and burstiness \cite{Holme2012PR, Holme2015EPJB, Masuda2020book, Karsai} may be a larger contributor to epidemic dynamics than concurrency. Further comparing the impact of concurrency on epidemic spreading and that of other network factors warrants future work.

Concurrency is not only relevant to sexually transmitted infections, as most studies of concurrency have focused on \cite{Masuda2021RoS,Aral2010CIDR,Kretzschmar2012AIDS,Lurie2010AIDS,Foxman2006STD,Sawers2013J.I.AIDS}, or even to epidemic spreading in general.  Other network dynamics such as opinion formation, synchronization, and information spreading and cascading on social networks may also be affected by concurrency. Our assumption that the partnership between any given pair of individuals can reoccur after it has disappeared is unrealistic in most cases of sexual partnerships. However, this assumption is considered to be natural for describing other types of dynamic social contacts such as face-to-face encounters and online communications, which underlie social dynamics apart from sexually transmitted infections. Investigating impacts of concurrency on these different social dynamics using the present models also warrants future work.

\section*{Acknowledgments}
M.O. acknowledges support from JSPS KAKENHI (under Grant No. JP21H01352). N.M. acknowledges support from AFOSR European Office (under Grant No. FA9550-19-1-7024), the Nakatani Foundation, the Sumitomo Foundation, and the Japan Science and Technology Agency (JST) Moonshot R$\&$D (under Grant No. JPMJMS2021).

\begin{appendices}

\section{Proof of Proposition \ref{2}}
\label{Appendix: A}

In model 2, random variable $y_i$ obeys a Bernoulli distribution with expected value $\mathbb E[y_i] = a/(a+b)$ because $y_i = 1$ with probability $a/(a+b)$ and $y_i = 0$ with probability $b/(a+b)$. Because $y_i$ and $y_j$, where $i\neq j$, are independent of each other, we obtain
\begin{align}
\mathbb E[k_i] & = \mathbb E\left[y_i\sum_{j=1}^N A_{ij}y_j\right] = \sum\limits_{\substack{j=1 \\ j\neq i}}^{N} A_{ij}\mathbb E[y_i]\mathbb E[y_j] = \overline{k}_i\left(\frac{a}{a+b}\right)^2
\end{align}
and
\begin{align}
\sigma^2[k_i] & = \mathbb E[k_i^2]-\left(\mathbb E[k_i]\right)^2 \nonumber\\
& = \mathbb E\left[\sum_{j=1}^N A_{ij}^2y_{i}^2y_{j}^2\right]+\mathbb E\left[\sum\limits_{\substack{j=1\\ j\neq i}}^N\sum\limits_{\substack{k=1\\ k\neq i,j}}^N A_{ij}A_{ik}y_{i}^2y_{j}y_{k}\right]-\left(\mathbb E[k_i]\right)^2\nonumber\\
& = \mathbb E\left[\sum_{j=1}^N A_{ij}y_{i}y_{j}\right]+\mathbb E\left[\sum\limits_{\substack{j=1\\ j\neq i}}^N\sum\limits_{\substack{k=1\\ k\neq i,j}}^N A_{ij}A_{ik}y_{i}y_{j}y_{k}\right]-\left(\mathbb E[k_i]\right)^2\nonumber\\
& = \overline{k}_i\left(\frac{a}{a+b}\right)^2 + \overline{k}_i(\overline{k}_i-1)\left(\frac{a}{a+b}\right)^3 - \left[\overline{k}_i\left(\frac{a}{a+b}\right)^2\right]^2\nonumber\\
& = \frac{\overline{k}_ia^2b}{(a+b)^3}\left(1+\frac{\overline{k}_ia}{a+b}\right).
\end{align}
We also obtain
\begin{align}
\mathbb E[\langle k\rangle] & = \mathbb E\left[\frac{2}{N} \sum_{i=1}^N \sum_{j=i+1}^{N} A_{ij}y_{i}y_{j}\right]\nonumber\\
& = \frac{2}{N} \sum_{i=1}^N \sum_{j=i+1}^{N} A_{ij} \mathbb E[y_i]\mathbb E[y_j]\nonumber\\
& = \frac{2M}{N}\left(\frac{a}{a+b}\right)^2
\end{align}
and
\begin{align}
\sigma^2[\langle k\rangle] & = \mathbb E[\langle k\rangle^2] - \left(\mathbb E[\langle k\rangle]\right)^2\nonumber\\
& = \mathbb E\left[\left(\frac{2}{N} \sum_{i=1}^N \sum_{j=i+1}^{N} A_{ij}y_{i}y_{j}\right)^2\right] - \left[\frac{2M}{N}\left(\frac{a}{a+b}\right)^2\right]^2\nonumber\\
& = \frac{4}{N^2}\mathbb E\left[\sum_{i=1}^N \sum_{j=i+1}^{N} A_{ij}^{2}y_{i}^{2}y_{j}^{2} + 2\sum_{i=1}^N \sum_{j=i+1}^{N} \sum_{k=j+1}^{N} A_{ij}A_{jk}y_{i}y_{j}^{2}y_{k} + 2\sum_{i=1}^N \sum_{j=i+1}^{N} \sum_{k=j+1}^{N} A_{ij}A_{ik}y_{i}^{2}y_{j}y_{k} \right. \nonumber\\
& \phantom{=\;\;}\left.+ 2\sum_{i=1}^N \sum_{k=i+1}^{N} \sum_{j=k+1}^{N} A_{ij}A_{kj}y_{i}y_{k}y_{j}^{2} + 2\sum_{i=1}^{N-1}\sum_{j=i+1}^N\sum\limits_{\substack{k>i\\k\neq j}}^{N-1}\sum\limits_{\substack{\ell>k\\\ell\neq i,j}}^N A_{ij}A_{k\ell}y_{i}y_{j}y_{k}y_{\ell}\right] - \left(\frac{2M}{N}\right)^2\left(\frac{a}{a+b}\right)^4\nonumber\\
& = \frac{4}{N^2}\left(\mathbb E\left[\sum_{i=1}^N \sum_{j=i+1}^{N} A_{ij}y_{i}y_{j}\right] + 2\mathbb E\left[\sum_{i=1}^N \sum_{j=i+1}^{N} \sum_{k=j+1}^{N} A_{ij}A_{jk}y_{i}y_{j}y_{k}\right] + 2\mathbb E\left[\sum_{i=1}^N \sum_{j=i+1}^{N} \sum_{k=j+1}^{N} A_{ij}A_{ik}y_{i}y_{j}y_{k}\right] \right.\nonumber\\
& \phantom{=\;\;}\left. + 2\mathbb E\left[\sum_{i=1}^N \sum_{k=i+1}^{N} \sum_{j=k+1}^{N} A_{ij}A_{kj}y_{i}y_{k}y_{j}\right] + 2\mathbb E\left[\sum_{i=1}^{N-1}\sum_{j=i+1}^N\sum\limits_{\substack{k>i\\k\neq j}}^{N-1}\sum\limits_{\substack{\ell>k\\\ell\neq i,j}}^N A_{ij}A_{k\ell}y_{i}y_{j}y_{k}y_{\ell}\right]\right) - \left(\frac{2M}{N}\right)^2\left(\frac{a}{a+b}\right)^4\nonumber\\
& = \frac{4}{N^2}\left[M\left(\frac{a}{a+b}\right)^2 + 2\left(\frac{a}{a+b}\right)^3\sum_{i=1}^N \sum_{j=i+1}^{N} \sum_{k=j+1}^{N} A_{ij}A_{jk} + 2\left(\frac{a}{a+b}\right)^3\sum_{i=1}^N \sum_{j=i+1}^{N} \sum_{k=j+1}^{N} A_{ij}A_{ik}\right.\nonumber\\
& \phantom{=\;\;}\left. + 2\left(\frac{a}{a+b}\right)^3\sum_{i=1}^N \sum_{k=i+1}^{N} \sum_{j=k+1}^{N} A_{ij}A_{kj} + 2\left(\frac{a}{a+b}\right)^4\sum_{i=1}^{N-1}\sum_{j=i+1}^N\sum\limits_{\substack{k>i\\k\neq j}}^{N-1}\sum\limits_{\substack{\ell>k\\\ell\neq i,j}}^N A_{ij}A_{k\ell}\right] - \left(\frac{2M}{N}\right)^2\left(\frac{a}{a+b}\right)^4\nonumber\\
& = \frac{4M}{N^2}\left(\frac{a}{a+b}\right)^2 + \frac{8}{N^2}\left(\frac{a}{a+b}\right)^3\left(\sum_{i=1}^N \sum_{j=i+1}^{N} \sum_{k=j+1}^{N} A_{ij}A_{jk} + \sum_{i=1}^N \sum_{j=i+1}^{N} \sum_{k=j+1}^{N} A_{ij}A_{ik} + \sum_{i=1}^N \sum_{k=i+1}^{N} \sum_{j=k+1}^{N} A_{ij}A_{kj}\right.\nonumber\\
& \phantom{=\;\;}\left.+\sum_{i=1}^{N-1}\sum_{j=i+1}^N\sum\limits_{\substack{k>i\\k\neq j}}^{N-1}\sum\limits_{\substack{\ell>k\\\ell\neq i,j}}^N A_{ij}A_{k\ell}\right) + \frac{8}{N^2}\left(\frac{a}{a+b}\right)^3\sum_{i=1}^{N-1}\sum_{j=i+1}^N\sum\limits_{\substack{k>i\\k\neq j}}^{N-1}\sum\limits_{\substack{\ell>k\\\ell\neq i,j}}^N A_{ij}A_{k\ell}\left(\frac{a}{a+b} - 1\right) - \left(\frac{2M}{N}\right)^2\left(\frac{a}{a+b}\right)^4\nonumber\\
& = \frac{4M}{N^2}\left(\frac{a}{a+b}\right)^2 + \frac{4M(M-1)}{N^2}\left(\frac{a}{a+b}\right)^3 - \frac{8}{N^2}\frac{a^3b}{(a+b)^4}\sum_{i=1}^{N-1}\sum_{j=i+1}^N\sum\limits_{\substack{k>i\\k\neq j}}^{N-1}\sum\limits_{\substack{\ell>k\\\ell\neq i,j}}^N A_{ij}A_{k\ell}\nonumber\\ 
& \quad - \left(\frac{2M}{N}\right)^2\left(\frac{a}{a+b}\right)^4.
\end{align}

\section{Proof of Proposition \ref{4}}
\label{Appendix: B}

For model 3, we obtain
\begin{align}
\mathbb E[k_i] & = \mathbb E\left[\sum_{j=1}^N A_{ij}(y_i+y_j-y_iy_j)\right] \nonumber\\
& = \sum\limits_{\substack{j=1 \\ j\neq i}}^{N} A_{ij}(\mathbb E[y_i]+\mathbb E[y_j]-\mathbb E[y_i]\mathbb E[y_j]) = \frac{\overline{k}_ia(a+2b)}{(a+b)^2}
\end{align}
and
\begin{align}
\sigma^2[k_i] & = \mathbb E[k_i^2]-\left(\mathbb E[k_i]\right)^2\nonumber\\
& = \mathbb E\left[\sum_{j=1}^N A_{ij}^2(y_i+y_j-y_iy_j)^2\right]+\mathbb E\left[\sum\limits_{\substack{j=1\\ j\neq i}}^N\sum\limits_{\substack{k=1\\ k\neq i,j}}^N A_{ij}A_{ik}(y_i+y_jy_k-y_iy_jy_k)\right]-\left(\mathbb E[k_i]\right)^2\nonumber\\
& = \mathbb E\left[\sum_{j=1}^N A_{ij}(y_i+y_j-y_iy_j)\right] + \mathbb E\left[\sum\limits_{\substack{j=1\\ j\neq i}}^N\sum\limits_{\substack{k=1\\ k\neq i,j}}^N A_{ij}A_{ik}(y_i+y_jy_k-y_iy_jy_k)\right]-\left(\mathbb E[k_i]\right)^2\nonumber\\
& = \frac{\overline{k}_ia(a+2b)}{(a+b)^2}+\frac{\overline{k}_i(\overline{k}_i-1)a(a^2+3ab+b^2)}{(a+b)^3}-\left[\frac{\overline{k}_ia(a+2b)}{(a+b)^2}\right]^2.
\end{align}
We also obtain
\begin{align}
\mathbb E[\langle k\rangle] & = \mathbb E\left[\frac{2}{N} \sum_{i=1}^N \sum_{j=i+1}^{N} A_{ij}\left(y_i + y_j - y_iy_j\right)\right]\nonumber\\
& = \frac{2}{N} \sum_{i=1}^N \sum_{j=i+1}^{N} A_{ij}\left(\mathbb E[y_i] + \mathbb E[y_j] - \mathbb E[y_i]\mathbb E[y_j]\right)\nonumber\\
& = \frac{2}{N} \sum_{i=1}^N \sum_{j=i+1}^{N} A_{ij}\left[\frac{2a}{a+b} - \left(\frac{a}{a+b}\right)^2\right]\nonumber\\
& = \frac{2M}{N}\cdot\frac{a(a+2b)}{\left(a+b\right)^2}
\end{align}
and
\begin{align}
\sigma^2[\langle k\rangle] & = \mathbb E[\langle k\rangle^2] - \left(\mathbb E[\langle k\rangle]\right)^2\nonumber\\
& = \mathbb E\left[\left(\frac{2}{N} \sum_{i=1}^N \sum_{j=i+1}^{N} A_{ij}\left(y_i + y_j - y_iy_j\right)\right)^2\right] - \left[\frac{2M}{N}\cdot\frac{a(a+2b)}{\left(a+b\right)^2}\right]^2\nonumber\\
& = \frac{4}{N^2} \mathbb E\left[\sum_{i=1}^N \sum_{j=i+1}^{N} A_{ij}^2\left(y_i+y_j-y_iy_j\right)^2 + 2\sum_{i=1}^N \sum_{j=i+1}^{N} \sum_{k=j+1}^{N} A_{ij}A_{jk}\left(y_i+y_j-y_iy_j\right)\left(y_j+y_k-y_jy_k\right)\right.\nonumber\\
& \phantom{=\;\;}\left.+ 2\sum_{i=1}^N \sum_{j=i+1}^{N} \sum_{k=j+1}^{N} A_{ij}A_{ik}\left(y_i+y_j-y_iy_j\right)\left(y_i+y_k-y_iy_k\right)\right.\nonumber\\ 
& \phantom{=\;\;}\left.+ 2\sum_{i=1}^N \sum_{k=i+1}^{N} \sum_{j=k+1}^{N} A_{ij}A_{kj}\left(y_i+y_j-y_iy_j\right)\left(y_k+y_j-y_ky_j\right)\right.\nonumber\\
& \phantom{=\;\;}\left.+ 2\sum_{i=1}^{N-1}\sum_{j=i+1}^N\sum\limits_{\substack{k>i\\k\neq j}}^{N-1}\sum\limits_{\substack{\ell>k\\\ell\neq i,j}}^N A_{ij}A_{k\ell}\left(y_i+y_j-y_iy_j\right)\left(y_k+y_\ell-y_ky_\ell\right)\right] - \left[\frac{2M}{N}\cdot\frac{a(a+2b)}{\left(a+b\right)^2}\right]^2\nonumber\\
& = \frac{4}{N^2} \mathbb E\left[\sum_{i=1}^N \sum_{j=i+1}^{N} A_{ij}\left(y_i+y_j-y_iy_j\right)+2\sum_{i=1}^N \sum_{j=i+1}^{N} \sum_{k=j+1}^{N} A_{ij}A_{jk}\left(y_j+y_iy_k-y_iy_jy_k\right)\right.\nonumber\\
& \phantom{=\;\;}\left.+2\sum_{i=1}^N \sum_{j=i+1}^{N} \sum_{k=j+1}^{N} A_{ij}A_{ik}\left(y_i+y_jy_k-y_iy_jy_k\right)+2\sum_{i=1}^N \sum_{k=i+1}^{N} \sum_{j=k+1}^{N} A_{ij}A_{kj}\left(y_j+y_iy_k-y_iy_jy_k\right)\right.\nonumber\\
& \phantom{=\;\;}\left.+2\sum_{i=1}^{N-1}\sum_{j=i+1}^N\sum\limits_{\substack{k>i\\k\neq j}}^{N-1}\sum\limits_{\substack{\ell>k\\\ell\neq i,j}}^N A_{ij}A_{k\ell}\left(y_iy_k+y_iy_\ell+y_jy_k+y_jy_\ell-y_iy_ky_\ell-y_jy_ky_\ell-y_iy_jy_k-y_iy_jy_\ell+y_iy_jy_ky_\ell\right)\right]\nonumber\\
& \quad \ - \left[\frac{2M}{N}\cdot\frac{a(a+2b)}{\left(a+b\right)^2}\right]^2\nonumber\\
& = \frac{4M}{N^2}\cdot\frac{a(a+2b)}{\left(a+b\right)^2}+\frac{8}{N^2}\sum_{i=1}^N \sum_{j=i+1}^{N} \sum_{k=j+1}^{N} A_{ij}A_{jk}\frac{a(a^2+3ab+b^2)}{\left(a+b\right)^3}\nonumber\\
& \quad \ +\frac{8}{N^2}\sum_{i=1}^N \sum_{j=i+1}^{N} \sum_{k=j+1}^{N} A_{ij}A_{ik}\frac{a(a^2+3ab+b^2)}{\left(a+b\right)^3}+\frac{8}{N^2}\sum_{i=1}^N \sum_{k=i+1}^{N} \sum_{j=k+1}^{N} A_{ij}A_{kj}\frac{a(a^2+3ab+b^2)}{\left(a+b\right)^3}\nonumber\\
& \quad \ +\frac{8}{N^2}\sum_{i=1}^{N-1}\sum_{j=i+1}^N\sum\limits_{\substack{k>i\\k\neq j}}^{N-1}\sum\limits_{\substack{\ell>k\\\ell\neq i,j}}^N A_{ij}A_{k\ell}\frac{a^2(a^2+4ab+4b^2)}{\left(a+b\right)^4} - \left[\frac{2M}{N}\cdot\frac{a(a+2b)}{\left(a+b\right)^2}\right]^2\nonumber\\
& = \frac{4M}{N^2}\cdot\frac{a(a+2b)}{\left(a+b\right)^2}+\frac{4M(M-1)}{N^2}\cdot\frac{a(a^2+3ab+b^2)}{\left(a+b\right)^3} - \frac{8}{N^2}\cdot\frac{ab^3}{\left(a+b\right)^4}\sum_{i=1}^{N-1}\sum_{j=i+1}^N\sum\limits_{\substack{k>i\\k\neq j}}^{N-1}\sum\limits_{\substack{\ell>k\\\ell\neq i,j}}^N A_{ij}A_{k\ell}\nonumber\\
& \quad \ - \left[\frac{2M}{N}\cdot\frac{a(a+2b)}{\left(a+b\right)^2}\right]^2.
\end{align}

\section{\label{Appendix: C}Proof of Proposition~\ref{6}}


Let $a_2$ and $a_3$ be the rate at which the state of a node changes from the $\ell$ to the $h$ state for models 2 and 3, respectively. Similarly, let $b_2$ and $b_3$ be the rate at which the state of a node changes from the $h$ to $\ell$ state for models 2 and 3, respectively. Because $q^*$ is the same between models 2 and 3 by assumption, Eq.~\eqref{eo23} holds true. We apply Eq.~\eqref{eo23} to Eqs.~\eqref{Di2} and \eqref{Di3} to obtain
\begin{align}
& \sigma^2_3[k_i] - \sigma^2_2[k_i]\nonumber\\
= & \overline{k}_i\frac{a_3(a_3+2b_3)}{\left(a_3+b_3\right)^2}+\overline{k}_i(\overline{k}_i-1)\frac{a_3(a_3^2+3a_3b_3+b_3^2)}{\left(a_3+b_3\right)^3}-\left[\overline{k}_i\frac{a_3(a_3+2b_3)}{\left(a_3+b_3\right)^2}\right]^2\nonumber\\
\quad & - \overline{k}_i\left(\frac{a_2}{a_2+b_2}\right)^2 - \overline{k}_i(\overline{k}_i-1)\left(\frac{a_2}{a_2+b_2}\right)^3 + \left[\overline{k}_i\left(\frac{a_2}{a_2+b_2}\right)^2\right]^2\nonumber\\
= & \overline{k}_i(\overline{k}_i-1)\frac{a_3(a_3^2+3a_3b_3+b_3^2)}{\left(a_3+b_3\right)^3}-\overline{k}_i(\overline{k}_i-1)\left(\frac{a_2}{a_2+b_2}\right)^3\nonumber\\
= & \overline{k}_i(\overline{k}_i-1)\frac{a_3(a_3^2+3a_3b_3+b_3^2)}{\left(a_3+b_3\right)^3}-\overline{k}_i(\overline{k}_i-1)\frac{\sqrt{a_3^3(a_3+2b_3)^3}}{\left(a_3+b_3\right)^3}\nonumber\\
= & \frac{\overline{k}_i(\overline{k}_i-1)}{\left(a_3+b_3\right)^3}\left[a_3(a_3^2+3a_3b_3+b_3^2)-\sqrt{a_3^3(a_3+2b_3)^3}\right].
\end{align}
Therefore, $\sigma^2_3[k_i] < \sigma^2_2[k_i]$ if and only if $b_3<(1+\sqrt{2})a_3$, which is equivalent to $\frac{1}{2}<q^*\leq 1$.

By applying Eq.~\eqref{eo23} to Eqs.~\eqref{D2} and \eqref{D3} and using Eq.~\eqref{combination}, we obtain
\begin{align}
& \sigma^2_2[\langle k\rangle]-\sigma^2_3[\langle k\rangle]\nonumber\\
= & \frac{4M(M-1)}{N^2}\left(\frac{a_2}{a_2+b_2}\right)^3 - \frac{8}{N^2}\left(\frac{a_2}{a_2+b_2}\right)^3\frac{b_2}{a_2+b_2}\sum_{i=1}^{N-1}\sum_{j=i+1}^N\sum\limits_{\substack{k>i\\k\neq j}}^{N-1}\sum\limits_{\substack{\ell>k\\\ell\neq i,j}}^N A_{ij}A_{k\ell}\nonumber\\
& -\frac{4M(M-1)}{N^2}\cdot\frac{a_3(a_3^2+3a_3b_3+b_3^2)}{(a_3+b_3)^3} + \frac{8}{N^2}\cdot\frac{a_3b_3^3}{(a_3+b_3)^4}\sum_{i=1}^{N-1}\sum_{j=i+1}^N\sum\limits_{\substack{k>i\\k\neq j}}^{N-1}\sum\limits_{\substack{\ell>k\\\ell\neq i,j}}^N A_{ij}A_{k\ell}\nonumber\\
= & \frac{4M(M-1)}{N^2}\cdot\frac{\sqrt{a_3^3(a_3+2b_3)^3}}{\left(a_3+b_3\right)^3}-\frac{8}{N^2}\cdot \frac{\sqrt{a_3^3(a_3+2b_3)^3}}{\left(a_3+b_3\right)^3}\cdot\frac{a_3+b_3-\sqrt{a_3(a_3+2b_3)}}{a_3+b_3}\sum_{i=1}^{N-1}\sum_{j=i+1}^N\sum\limits_{\substack{k>i\\k\neq j}}^{N-1}\sum\limits_{\substack{\ell>k\\\ell\neq i,j}}^N A_{ij}A_{k\ell}\nonumber\\
& \quad \ -\frac{4M(M-1)}{N^2}\cdot\frac{a_3(a_3^2+3a_3b_3+b_3^2)}{\left(a_3+b_3\right)^3} +\frac{8}{N^2}\cdot\frac{a_3b_3^3}{\left(a_3+b_3\right)^4}\sum_{i=1}^{N-1}\sum_{j=i+1}^N\sum\limits_{\substack{k>i\\k\neq j}}^{N-1}\sum\limits_{\substack{\ell>k\\\ell\neq i,j}}^N A_{ij}A_{k\ell}\nonumber\\
= & \frac{4M(M-1)}{N^2}\cdot\frac{\left(a_3+b_3\right)\left[\sqrt{a_3^3(a_3+2b_3)^3}-a_3(a_3^2+3a_3b_3+b_3^2)\right]}{\left(a_3+b_3\right)^4}\nonumber\\
& - \frac{8}{N^2}\sum_{i=1}^{N-1}\sum_{j=i+1}^N\sum\limits_{\substack{k>i\\k\neq j}}^{N-1}\sum\limits_{\substack{\ell>k\\\ell\neq i,j}}^N A_{ij}A_{k\ell}\frac{\left(a_3+b_3\right)\sqrt{a_3^3(a_3+2b_3)^3}-a_3^2(a_3+2b_3)^2-a_3b_3^3}{\left(a_3+b_3\right)^4}\nonumber\\
= & \frac{8}{N^2}\left[\frac{M(M-1)}{2}-\sum_{i=1}^{N-1}\sum_{j=i+1}^N\sum\limits_{\substack{k>i\\k\neq j}}^{N-1}\sum\limits_{\substack{\ell>k\\\ell\neq i,j}}^N A_{ij}A_{k\ell}\right]\frac{\sqrt{a_3^3(a_3+2b_3)^3}-a_3(a_3^2+3a_3b_3+b_3^2)}{\left(a_3+b_3\right)^3}.
\end{align}
Therefore, $\sigma^2_3[\langle k\rangle] < \sigma^2_2[\langle k\rangle]$ if and only if $b_3<(1+\sqrt{2})a_3$, which is equivalent to $\frac{1}{2}<q^*\leq 1$.


\section{\label{sec:Masaki}Epidemic threshold when the network switches rapidly}

In this section, we evaluate the epidemic threshold for models 1, 2, and 3 under the assumption that the time-independent networks at different times are independent of each other in each model.  


%
%

We denote the identity and the zero matrices by $I$ and~$O$, respectively. 
A real matrix~$A$ (or a vector as its special case) is said to be nonnegative,
denoted by $A\geq 0$, if all the entries of~$A$ are nonnegative. If all the entries of~$A$ are positive, then $A$ is said to be positive. We say that $A\leq B$, where $A$ and $B$ are of the same dimension, if $B-A\geq 0$. A square matrix~$A$ is said to be Metzler if all its off-diagonal entries are nonnegative~\cite{Farina2000}. If $A$ is Metzler, it holds true that $e^{At} \geq
0$ for all $t\geq 0$ \cite{Farina2000}. For a Metzler matrix~$A$, the maximum real part of the eigenvalues of~$A$ is denoted
by~$\lambda_{\max}(A)$. For any matrix $A$, the spectral radius is the largest absolute value of its eigenvalues and denoted by $\rho(A)$.

The SIS model is a continuous-time Markov process with $2^N$ possible states
\cite{VanMieghem2009a,Pastor2015RMP,Kiss2017book} and
has a unique absorbing state in which all the $N$ nodes are susceptible. Because this absorbing state is reachable from any other state, the dynamics of the SIS model reaches the disease-free absorbing equilibrium in finite time with probability one. 


\subsection{A lower bound on the decay rate for subgraphs} \label{sec1}

We refer to the directed edges of a given network $\mathcal{\overline G}$ as $\mathcal{\overline E}=\{e_1, \ldots, e_M\}$, where the $\ell$th edge ($1\le \ell\le M$) is represented by $e_\ell = (i_\ell, j_\ell)$, i.e., the edge is directed from node $v_{i_\ell}$ to node $v_{j_\ell}$. Define the incidence matrix $\overline{C}\in \mathbb{R}^{N\times M}$ of the network $\mathcal{\overline G}$ by \cite{Wilson1972-2010book,Newman2018book}
\begin{equation}
\overline{C}_{i\ell} = \begin{cases}
1, & \mbox{if $j_\ell = i$}, 
\\
-1, & \mbox{if $i_\ell = i$}, 
\\
0, & \mbox{otherwise},
\end{cases}
\end{equation}
for all $i\in \{1, \dotsc, N\}$ and $\ell \in \{1, \dotsc, M\}$.
We also define the non-backtracking matrix $\overline{H}\in \mathbb{R}^{M\times M}$ of $\mathcal{\overline G}$ by \cite{Hashimoto1989AdvStudPureMath,Alon2007CommContempMath}
\begin{equation}\label{non-back_overlineH}
\overline{H}_{\ell m} = \begin{cases}
1, & \mbox{if $j_\ell = i_m$ and $j_m \neq i_\ell$}, 
\\
0, & \mbox{otherwise.}
\end{cases}
\end{equation}


We prove the following corollary about the SIS spreading processes taking place in a subgraph of $\mathcal G$.

\begin{corollary}\label{cor:subgraphSIS}
  Let $\breve{\mathcal E}$ be a subset of $\mathcal{\overline E}$. Consider the SIS model over the network $\breve{\mathcal G} = (\mathcal V, \breve{\mathcal E})$ with the infection rate $\beta > 0$ and recovery rate $\mu > 0$. Define matrix $\breve H\in \mathbb{R}^{M\times M}$ by 
\begin{equation}
\breve H_{\ell m} = \begin{cases}
1, & \mbox{if $e_\ell \in \breve{\mathcal E}$, {$j_\ell = i_m$}, and {$j_m \neq i_\ell$}}, 
\\
0, & \mbox{otherwise,}
\end{cases}
\end{equation}
We also define the diagonal matrix 
  \begin{equation}
    \breve\Xi =\text{diag}(\breve\xi_1, \ldots, \breve\xi_M),
  \end{equation}
  where 
  \begin{equation}
    \breve\xi_\ell = \begin{cases}
      1,&\mbox{if $e_\ell \in \breve{\mathcal{E}}$}, 
      \\
      0,&\mbox{otherwise}
    \end{cases}
  \end{equation}
  for each $\ell \in \{1, \dotsc, M\}$.
Then, we obtain
\begin{equation}
  \frac{d}{dt}\begin{bmatrix}
  \bm p\\ \bm q
  \end{bmatrix}
  \leq 
  \breve{\mathcal A} \begin{bmatrix}
  \bm p\\ \bm q
  \end{bmatrix},  
  \end{equation}  
where $\bm p(t) = \left[p_1(t), \ldots, p_N(t)\right]^{\top}$, $p_i(t) = \mathbb E[x_i(t)]$ is the probability that node $v_i$ is infectious at time $t$, $^{\top}$ represents the transposition,
$\bm q(t) = \left[ q_{i_1 j_1}(t), \ldots, q_{i_M j_M}(t) \right]^{\top}$, $q_{ij}(t)=\mathbb E[x_i(t)(1-x_j(t)]$ is the joint probability that $v_i$ is infectious and node $v_j$ is susceptible at time $t$, and the $(N+M)\times (N+M)$ matrix $\breve{\mathcal A}$ is defined by 
\begin{equation}
  \breve{\mathcal A} = \begin{bmatrix}
    -\mu I & \beta \overline{C}_+\breve\Xi
    \\
    \mu \overline{C}_-^\top & \beta \breve H^\top - \beta \breve\Xi -2\mu I
  \end{bmatrix},
\end{equation}
where $\overline{C}_+ = \max(\overline{C}, 0)$ and $\overline{C}_- = \max(-\overline{C}, 0)$ denote the positive and negative parts of the incidence matrix~$\overline{C}$, respectively.
\end{corollary}

\begin{proof}
We adapt Eq. (3.20) in Ref.~\cite{Masuda2020ImaJApplMath} to the case in which the network is defined by $\breve{\mathcal{G}}$, the infection rate is independent of edges, and the recovery rate is independent of nodes. Specifically, by replacing $B'$, $D$, $D_1'$, $D_2'$ in Eq. (3.20) in Ref.~\cite{Masuda2020ImaJApplMath} by $\beta\breve\Xi$, $\mu I$, $\mu I$, $\mu I$, respectively, we obtain
\begin{equation}\label{eq:pSIS}
\frac{d}{dt}\begin{bmatrix}
\bm p\\ \bm q
\end{bmatrix}
=
\mathcal A \begin{bmatrix}
\bm p\\ \bm q
\end{bmatrix} - \begin{bmatrix}
\bm 0 \\ \bm \epsilon
\end{bmatrix},
\end{equation}
where $\bm \epsilon(t)$ is entry-wise nonnegative for every $t\geq 0$.
Equation \eqref{eq:pSIS} implies 
\begin{equation}
  \frac{d}{dt}\begin{bmatrix}
  \bm p\\ \bm q
  \end{bmatrix}
  \leq 
  {\mathcal A} \begin{bmatrix}
  \bm p\\ \bm q
  \end{bmatrix}, 
  \end{equation} 
where matrix $\mathcal{A}$ is given by
\begin{equation}\label{eq:defMathcalA}
  \mathcal A = \begin{bmatrix}
  -\mu I & \beta \overline{C}_+\breve\Xi
  \\
  \mu \overline{C}_-^\top & \beta \overline{H}^\top {\breve\Xi} - \beta \breve\Xi - 2\mu I
  \end{bmatrix}. 
  \end{equation}
Therefore, to complete the proof, it is sufficient to show 
\begin{equation}\label{eq:sufftoshow2}
  \overline{H}^\top \breve\Xi = \breve H^\top, 
\end{equation}
which proves $\mathcal A= \breve{\mathcal A}$. 

Let us show Eq.~\eqref{eq:sufftoshow2}. For any $\ell, m \in \{1, \dotsc, M\}$, one obtains
\begin{align}
  [\overline{H}^\top \breve\Xi]_{\ell m}
  &=
  \sum_{k=1}^M [\overline{H}^\top]_{\ell k} [\breve\Xi]_{km}\notag
  \\
  &=
   [\overline{H}^\top]_{\ell m} [\breve\Xi]_{mm}\notag
   \\
   &=
   \overline{H}_{m \ell} \breve\xi_m\notag
   \\
   &=
   \begin{cases}
     1,&\mbox{if $e_m \in \breve{\mathcal E}$, $j_m=i_\ell$, and $j_\ell \neq i_m$}, 
     \\
     0,&\mbox{otherwise},
   \end{cases} \notag
   \\
   &=
   [\breve H^\top]_{\ell m}, 
\end{align}
which proves \eqref{eq:sufftoshow2}. 
\end{proof}

\subsection{A lower bound on the decay rate for temporal networks}

Let $\tilde{\mathcal G}_1 = (\mathcal V, \tilde{\mathcal E}_1)$, $\ldots$, $\tilde{\mathcal G}_L = (\mathcal V, \tilde{\mathcal E}_L)$ be directed and unweighted networks having the common node set $\mathcal V = \{v_1, \dotsc, v_N\}$. 
Let $\{\sigma_k\}_{k=0}^\infty$ be independent and identically distributed random variables following a probability distribution on the set $\{1, \dotsc, L\}$.  
Variable $\sigma_k$ indexes the $k$th network to be used. Let $h>0$ be arbitrary. For a real number $x$, let $\lfloor x \rfloor$ denote the maximum integer that does not exceed~$x$. 
Define the stochastic temporal network $\mathcal G$ by
\begin{equation}\label{eq:defmathcalG}
  \mathcal G(t) = \tilde{\mathcal G}_{\sigma_{\lfloor t/h \rfloor}}
\end{equation}
for all $t \geq 0$. Equation~\eqref{eq:defmathcalG} implies that $\mathcal G(kh+\tau) = \tilde{\mathcal G}_{\sigma_{k}}$
for all $k\geq 0$ and $\tau \in [0, h)$. 
In the temporal network models 1, 2, and 3 used in the main text, the states of individual edges or nodes independently repeat flipping in continuous time in a Poissonian manner, which induces network switching. Therefore, the duration for a single time-independent network, $h$, is not a constant. Furthermore, the time-independent networks before and after a single flipping of an edge's or node's state are not independent of each other. For these two reasons, the present $\mathcal G$ is not the same as models 1, 2, or 3. However, one can enforce the present $\mathcal G$ and any of models 1, 2, and 3 to have the same probability that each type of time-independent network, $\tilde{\mathcal G}_{\ell}$, occurs by setting an appropriate probability distribution for $\{ \tilde{\mathcal G}_1, \ldots, \tilde{\mathcal G}_L \}$ for the present $\mathcal{G}$. We assume such a probability distribution in the following text. We consider the stochastic SIS model taking place in $\mathcal G$.

\begin{definition}
  \label{defn:new:temporal}
  The \textit{decay rate} of the SIS model over $\mathcal G$ is defined by
  \begin{equation}
  \gamma = - \limsup_{t\to\infty} \frac{\log \sum_{i=1}^N \mathbb E[p_i(t)]}{t},
  \end{equation}
  where all nodes are assumed to be infected at $t=0$.
\end{definition}

\color{black}
Definition~\ref{defn:new:temporal} states that $\sum_{i=1}^N p_i(t)$, which is equal to the expected number of infected nodes at time $t$, roughly decays exponentially in time in proportion to $e^{-\gamma t}$. Because the number of infected nodes always becomes zero in finite time, the SIS model on networks always has a positive decay rate. In fact, exact computation of the decay rate is computationally demanding because the decay rate is equal to the modulus of the largest real part of the non-zero eigenvalues of a $2^N\times 2^N$ matrix representing the
infinitesimal generator of the Markov chain~\cite{VanMieghem2009a}. Therefore, bounds of the decay rate that only require computation of much smaller matrices are available~\cite{Ganesh2005,Chakrabarti2008,VanMieghem2009a,Preciado2014,Ogura2018SystControlLett,Masuda2020ImaJApplMath}. Here we develop such a bound for temporal network $\mathcal G$.

\color{black}
Let us label the set of time-aggregated edges, $\overline{\mathcal E} = \bigcup_{\ell=1}^L \tilde{\mathcal E}_\ell$, as $\overline{\mathcal E} = \{e_1, \dotsc, e_M\}$. Define the time-aggregated graph $\overline{\mathcal G} = (\mathcal V, \overline{\mathcal E})$. Let $\overline C \in \mathbb{R}^{N\times M}$ denote the incidence matrix of $\overline{\mathcal G}$. Although the notations $\overline{\mathcal G}$, $\overline{\mathcal E}$, and $\overline C$ are common to those used in section \ref{sec1}, this should not arise confusions in the following text. For each $t\geq 0$, let $\mathcal E(t)$ denote the set of edges of the network at time $t$, i.e.,~$\mathcal G(t)$. In other words, $\mathcal E(t) = \tilde{\mathcal E}_{\sigma_{\lfloor t/h \rfloor}}$. We also define the matrix $H(t)\in \mathbb{R}^{M\times M}$ by
\begin{equation}\label{eq:def:tildeH}
  H(t)_{\ell m} = \begin{cases}
1 & \mbox{if $e_\ell \in \mathcal E(t)$, $j_\ell = i_m$, and $j_m \neq i_\ell$}, 
\\
0 & \mbox{otherwise,}
\end{cases}
\end{equation}
where $\ell, m \in \{1, \dotsc, M\}$,
and the diagonal matrix
  \begin{equation} \label{diagonalxi}
    \Xi(t) =\text{diag}(\xi_1(t), \ldots, \xi_M(t)),
  \end{equation}
  where 
  \begin{equation} \label{eachxi}
    \xi_\ell(t) = \begin{cases}
      1&\mbox{if $e_\ell \in {\mathcal{E}}(t)$}, 
      \\
      0&\mbox{otherwise,}
    \end{cases}
  \end{equation}
and $\ell \in \{1, \dotsc, M\}$. Note that $H(t)$ is similar to but different from the non-backtracking matrix of network $\mathcal{G}(t)$ because the definition of $H(t)$ does not require $e_m \in \mathcal{E}(t)$.

The following proposition gives an upper bound on the decay rate of the SIS model over the temporal network $\mathcal G$. 

\begin{proposition}\label{thm:temporal}
Define an $(N+M)\times (N+M)$ random matrix~$\mathcal A(t)$ by 
\begin{equation}\label{eq:defMathcalA:temporal}
\mathcal A(t) = \begin{bmatrix}
-\mu I & \beta \overline C_+\Xi(t) 
\\
\mu \overline C_-^\top & \beta H(t)^\top - \beta \Xi(t) - 2\mu I
\end{bmatrix}. 
\end{equation} 
Let 
\begin{equation}
  \mathcal F = \mathbb E[e^{h\mathcal A(0)}]. 
\end{equation}
Then, the decay rate of the SIS model over the temporal network~$\mathcal G$ is greater than or equal to $-h^{-1}\log \rho(\mathcal F)$. 
\end{proposition}

\begin{proof} 
  From Corollary~\ref{cor:subgraphSIS}, we obtain 
  \begin{equation}\label{eq:pSIS:temporal}
  \frac{d}{dt}\begin{bmatrix}
  \bm p(t)\\ \bm q(t)
  \end{bmatrix}
  \leq 
  \mathcal A(t) \begin{bmatrix}
  \bm p(t)\\ \bm q(t)
  \end{bmatrix}. 
  \end{equation}
  By combining Eq.~\eqref{eq:pSIS:temporal} and the definition of the temporal network given by \eqref{eq:defmathcalG}, we obtain
  \begin{equation}\label{eq:pSIS:temporal:tau}
    \frac{d}{d\tau}\begin{bmatrix}
    \bm p(kh+\tau)\\ \bm q(kh+\tau)
    \end{bmatrix}
    \leq 
    \mathcal A(kh) \begin{bmatrix}
    \bm p(\tau)\\ \bm q(\tau)
    \end{bmatrix}
    \end{equation}
for all $k=0, 1, \ldots$ and $\tau \in [0, h)$, which implies 
\begin{equation}\label{eq:pSIS:temporal:tau2}
\begin{bmatrix}
    \bm p((k+1)h)\\ \bm q((k+1)h)
    \end{bmatrix}
      \leq 
      e^{h\mathcal A (kh)}
      \begin{bmatrix}
        \bm p(kh)\\ \bm q(kh)
        \end{bmatrix}. 
\end{equation}
Because the random variables $\{\sigma_k\}_{k=0}^\infty$ are independently and identically distributed, we take the expectation with respect to $\sigma$ in Eq.~\eqref{eq:pSIS:temporal:tau2} to obtain
\begin{equation}
\begin{bmatrix}
    \mathbb E[\bm p((k+1)h)]\\ \mathbb E[\bm q((k+1)h)]
    \end{bmatrix}
      \leq 
      \mathcal F 
      \begin{bmatrix}
        \mathbb E[\bm p(kh)]\\ \mathbb E[\bm q(kh)]
        \end{bmatrix}, 
\end{equation}
which implies 
\begin{equation}\label{eq:epkhqkh}
\begin{bmatrix}
    \mathbb E[\bm p(kh)]\\ \mathbb E[\bm q(kh)]
    \end{bmatrix}
      \leq 
      \mathcal F^k 
      \begin{bmatrix}
        \bm p(0)\\ \bm q(0)
        \end{bmatrix}. 
\end{equation}

Now, let $\mathcal S \subset \mathbb{R}^{(N+M)\times (N+M)}$ denote the support of the random matrix $\mathcal A(0)$. Inequality \eqref{eq:pSIS:temporal:tau} shows that, for each sample path, there exists a matrix $X_{k}\in \mathcal S$ dependent on $k$ such that 
\begin{equation}\label{eq:kh->kh+tau}
  \begin{bmatrix}
    \bm p(kh+\tau)\\ \bm q(kh+\tau)
    \end{bmatrix}\leq e^{\tau X_{k}}
    \begin{bmatrix}
      \bm p(kh)\\ \bm q(kh)
      \end{bmatrix}
\end{equation}
for all $k=0, 1, \ldots$ and $\tau \in [0, h)$. Because $\mathcal S$ is finite, the maximum 
  $\Gamma = \max_{\substack{0\leq \tau\leq h, X\in \mathcal S}}\Norm{e^{\tau X}}$
exists and is finite. 
{Therefore, Eq.~\eqref{eq:kh->kh+tau} implies
\begin{equation}\label{eq:normkh->kh+tau}
  \Norm{\begin{bmatrix}
    \bm p(kh+\tau)\\ \bm q(kh+\tau)
    \end{bmatrix}}
    \leq 
    \norm{e^{\tau X_{k}}}
    \Norm{\begin{bmatrix}
      \bm p(kh)\\ \bm q(kh)
      \end{bmatrix}}
      \leq 
    \Gamma
    \Norm{\begin{bmatrix}
      \bm p(kh)\\ \bm q(kh)
      \end{bmatrix}}
\end{equation}
with probability one. 
By combining Eqs.~\eqref{eq:epkhqkh} and~\eqref{eq:normkh->kh+tau}, we obtain} 
\begin{equation}\label{eq:normkh+tau}
  \norm{
    \mathbb E[\bm p(kh+\tau)]} \leq \Gamma \norm{\mathcal F^k} 
    \Norm{\begin{bmatrix}
      \bm p(0)\\ \bm q(0)
      \end{bmatrix}}. 
\end{equation}
Equation~\eqref{eq:normkh+tau} implies that
\begin{equation}\label{eq:E32,new}
    \| \mathbb E[\bm p(kh+\tau)] \| 
    \leq c\|\mathcal F^k\|
\end{equation}
for a constant $c>0$. 
Because 
\begin{align}
    \limsup_{t\to\infty}\frac{\log \sum_{i=1}^N \mathbb E[p_i(t)]}{t}
    &\leq 
    \limsup_{t\to\infty}\frac{\log \|\mathbb E[\bm p(t)]\|}{t}\nonumber
    \\
    &=
    \lim_{k\to\infty} \sup_{\ell \geq k} \max_{0\leq \tau < h}\frac{\log \|\mathbb E[\bm p(\ell h+\tau)]\|}{\ell h+\tau}\nonumber
    \\
    &\leq 
    \frac{1}{h}\lim_{k\to\infty} \sup_{\ell \geq k} \max_{0\leq \tau < h}\frac{\log \|\mathbb E[\bm p(\ell h+\tau)]\|}{\ell}, 
\end{align}
inequality~\eqref{eq:E32,new} implies that 
\begin{align}\label{eq:logrho}
    \limsup_{t\to\infty}\frac{\log \sum_{i=1}^N \mathbb E[p_i(t)]}{t}
    &\leq 
    \frac{1}{h}\lim_{k\to\infty} \sup_{\ell \geq k} \frac{\log\|\mathcal F^\ell\| +\log c}{\ell}\nonumber
    \\
    &=
    \frac{1}{h}\limsup_{k\to\infty}\frac{\log\|\mathcal F^k\|}{k}\nonumber
    \\
    &=
    \frac{1}{h}\log \rho(\mathcal F).
\end{align}
In the last equality in Eq.~\eqref{eq:logrho}, we used the Gelfand's formula \cite[Corollary 5.6.14]{Horn1990}.
\end{proof}

  
  
\subsection{Epidemic threshold for small $h$ }

In this section, we consider the case in which the switching interval, $h$, is sufficiently small. 
The first-order expansion of matrix $\mathcal F$ with respect to $h$ yields 
\begin{equation}
  \mathcal F = I + \mathbb E[h {\mathcal A(0)}] + O(h^2) = I + h \mathbb E[{\mathcal A(0)}] + O(h^2). 
\end{equation}
Let us assume that matrix $\mathbb E[{\mathcal A(0)}]$ is irreducible. Let $\lambda_{\max}$ denote the real eigenvalue of $\mathbb E[{\mathcal A(0)}]$ with the largest real part. Then, by the Perron-Frobenius theorem and the differentiability of the eigenvalues of matrices \cite[Theorem 6.3.12]{Horn1990}, we obtain 
\begin{equation}
  \rho(\mathcal F) = 1 + h \lambda_{\max} + O(h^2). 
\end{equation}
Therefore, $h^{-1}\log \rho(\mathcal F) 
= \lambda_{\max}=O(h)$.   
The combination of this asymptotic evaluation and Proposition~\ref{thm:temporal} suggests that the epidemic threshold is the value of $\beta/\mu$ at which $\lambda_{\text{max}} = 0$.

We rewrite the epidemic threshold in terms of the spectral radius of the relevant matrices. As in the proof of Corollary 5 in Ref.~\cite{Masuda2020ImaJApplMath}, we rewrite the random matrix $\mathcal A(0)$ as 
\begin{equation}
  \mathcal A(0) = R + \mathcal P
  \end{equation}
  with a constant matrix
  \begin{equation}
  R = \begin{bmatrix}
  -\mu I & O
  \\
  \mu \bar C_{-}^\top & -\beta I - 2\mu I
  \end{bmatrix}
  \end{equation}
  and a random matrix
  \begin{equation}
  \mathcal P  = \beta \begin{bmatrix}
  O & \overline C_+\Xi(0)
  \\
  O & I + H(0)^\top - \Xi(0)
  \end{bmatrix}.
  \end{equation}
Then, we obtain the decomposition 
\begin{equation}
  \mathbb E[\mathcal A(0)] = R + P, 
\end{equation}
where 
  \begin{equation}
  P  = \mathbb E[\mathcal P] = \beta \begin{bmatrix}
  O & \overline C_+\mathbb E\bigl[\Xi(0)\bigr]
  \\
  O & I +\mathbb E\bigl[H(0)\bigr]^\top - \mathbb E\bigl[\Xi(0)\bigr] 
  \end{bmatrix}.
  \end{equation}
Matrix~$R$ is Metzler, and all the eigenvalues of $R$ have negative real parts. Matrix~$P$ is nonnegative because each diagonal of $\Xi(0)$ is a $\{0, 1\}$-valued random variable. Therefore, Theorem~2.11 in Ref.~\cite{Damm2003LinAlgItsAppl} implies that $\lambda_{\max} < 0$ if and only if $\rho(R^{-1} P) < 1$. Because
\begin{equation}
R^{-1} P = 
\begin{bmatrix}
O & -\dfrac \beta \mu \overline C_+\mathbb E\bigl[\Xi(0)\bigr] 
\\ 
O & -\dfrac \beta{\beta + 2\mu}\left(I + \mathbb E\bigl[H(0)\bigr]^\top- \mathbb E[\Xi(0)] + \overline C_{-}^\top \overline C_+\mathbb E\bigl[\Xi(0)\bigr] \right)
\end{bmatrix}, 
\end{equation}
one obtains
\begin{equation}\label{eq:rhoR-1P:temporal}
\rho(R^{-1}  P) = \frac{\beta}{\beta+2\mu}\rho\left(I + \mathbb E\bigl[H(0)\bigr]^\top- \mathbb E[\Xi(0)] + \overline C_-^\top \overline C_+\mathbb E\bigl[\Xi(0)\bigr] \right). 
\end{equation}
Equation \eqref{eq:rhoR-1P:temporal} implies that $\rho(R^{-1} P) < 1$ if and only if 
\begin{equation}\label{eq:et}
\frac{\beta}{\mu } < \frac{2}{\rho\left(I - \mathbb E[\Xi(0)] + \mathbb E\bigl[\Xi(0)\bigr] \overline C_+^\top \overline C_{-}  + \mathbb E\bigl[H(0)\bigr]\right) - 1}. 
\end{equation}

The right-hand side of Eq. \eqref{eq:et} gives the epidemic threshold, which we can further as follows. Equations \eqref{eq:def:tildeH}, \eqref{diagonalxi}, and \eqref{eachxi} imply that $H(t) = \Xi(t) \overline{H}$, where $\overline{H}$ is the non-backtracking matrix of the time-aggregated graph. (Note that this notation should not cause confusions with the definition of $\overline{H}$ given by Eq.~\eqref{non-back_overlineH}.) Therefore, we rewrite Eq.~\eqref{eq:et} as 
\begin{equation}\label{eq:et:new}
  \frac{\beta}{\mu } < \frac{2}{
    \rho\left(I - \mathbb E[\Xi(0)] + \overline C_{+}^\top  \mathbb E\bigl[\Xi(0)\bigr] \overline C_{-}  + \mathbb E\bigl[\Xi(0) \bigr] \overline H \right) - 1}. 
\end{equation}
Therefore, the epidemic threshold, i.e., the right-hand side of Eq.~\eqref{eq:et:new}, only depends on the expectation of $\Xi(0)$, i.e., $\mathbb E[\Xi(0)]$. We keep $\mathbb E[\Xi(0)]$ constant across the different models to inspect the effect of different amounts of concurrency on epidemic spreading under the condition that the aggregate network is the same across the comparisons. Therefore, with this analysis, we do not find a difference in the epidemic threshold across our different models.

\end{appendices}



\begin{thebibliography}{10}
\expandafter\ifx\csname url\endcsname\relax
  \def\url#1{\texttt{#1}}\fi
\expandafter\ifx\csname urlprefix\endcsname\relax\def\urlprefix{URL }\fi
\expandafter\ifx\csname href\endcsname\relax
  \def\href#1#2{#2} \def\path#1{#1}\fi

\bibitem{Newman2018book}
Newman, M. E. J. (2018). \textit{Networks}, (2nd ed.), Oxford University Press, Oxford.

\bibitem{Pastor2015RMP}
Pastor-Satorras, R., Castellano, C., Mieghem, P. V. \& Vespignani, A. (2015) Epidemic processes in complex networks. \textit{Rev. Mod. Phys.}  87(3), 925--979.

\bibitem{Kiss2017book}
Kiss, I. Z., Miller, J. C. \& Simon, P. L. (2017) \textit{Mathematics of Epidemics on Networks}, Springer, Cham.

\bibitem{Barrat2008book}
Barrat, A., Barth\'{e}lemy, M. \& Vespignani, A. (2008) \textit{Dynamical Processes on Complex Networks}, Cambridge University Press, Cambridge.
  
\bibitem{Holme2012PR}
Holme, P. \& Saramäki, J. (2012) Temporal networks. \textit{Phys. Rep.} 519, 97--125.

\bibitem{Holme2015EPJB}
Holme, P. (2015) Modern temporal network theory: a colloquium. \textit{Eur. Phys. J. B} 88(9), 234.

\bibitem{Masuda2020book}
Masuda, N. \& Lambiotte, R. (2020) \textit{A Guide to Temporal Networks}, 2nd ed., World Scientific, Singapore.

\bibitem{Bansal2010JBD}
Bansal, S., Read, J., Pourbohloul, B. \& Meyers, L. A. (2010) The dynamic nature of contact networks in infectious disease epidemiology. \textit{J. Biol. Dyn.} 4(5), 478--489.

\bibitem{Masuda2013F1000PR}
Masuda, N. \& Holme, P. (2013) Predicting and controlling infectious disease epidemics using temporal networks. \textit{F1000Prime Rep.} 5, 6.

\bibitem{Moody2002SF}
Moody, J. (2002) The importance of relationship timing for diffusion. \textit{Soc. Forces} 81(1), 25--56.

\bibitem{Moody}
Moody, J. \& Benton, R. A. (2016) Interdependent effects of cohesion and concurrency for epidemic potential. \textit{Ann. Epidemiol.} 26(4), 241--248.

\bibitem{Morris1995}
Morris,M. \& Kretzschmar,M. (1995) Concurrent partnerships and transmission dynamics in networks. \textit{Soc. Netw.} 17(3-4), 299--318.

\bibitem{Kretz1996}
Kretzschmar, M. \& Morris, M. (1996) Measures of concurrency in networks and the spread of infectious disease. \textit{Math. Biosci.} 133(2), 165--195.

\bibitem{Goodreau}
Goodreau, S. M. (2011) A decade of modelling research yields considerable evidence for the importance of concurrency: a response to Sawers and Stillwaggon. \textit{J. Int. AIDS Soc.} 14(1), 12.

\bibitem{Masuda2021RoS}
 Masuda, N.,Miller, J. C. \& Holme, P. (2021) Concurrency measures in the era of temporal network epidemiology: a review. \textit{J. R. Soc. Interface} 18(179), 20210019.
 
\bibitem{Lee}
Lee, E., Moody, J. \& Mucha, P. J. (2019) Exploring concurrency and reachability in the presence of high temporal resolution. In: Holme, P. \& Saramäki, J. (editors), \textit{Temporal Network Theory}, Springer, Cham, pp. 129--145.

\bibitem{Morris2010PlosONE}
Morris, M., Epstein, H. \& Wawer, M. (2010) Timing is everything: international variations in historical sexual partnership concurrency and HIV prevalence. \textit{PLoS ONE} 5(11), e14092.

\bibitem{Morris1997}
Morris, M. \& Kretzschmar, M. (1997) Concurrent partnerships and the spread of HIV. \textit{AIDS} 11(5), 641--648.

\bibitem{Watts}
Watts, C. H. \& May, R. M. (1992) Concurrent partnerships and transmission dynamics in networks. \textit{Math. Biosci} 108(1), 89--104.

\bibitem{Aral2010CIDR}
Aral, S. O. (2010) Partner concurrency and the STD/HIV epidemic. \textit{Curr. Infect. Dis. Rep.} 12(2), 134--139.

\bibitem{Kretzschmar2012AIDS}
Kretzschmar, M. \& Cara\" el, M. (2012) Is concurrency driving HIV transmission in sub-Saharan African sexual networks? The significance of sexual partnership typology. \textit{AIDS Behav.} 16(7), 1746--1752.

\bibitem{Lurie2010AIDS}
Lurie, M. N. \& Rosenthal, S. (2010) Concurrent partnerships as a driver of the HIV epidemic in sub-Saharan Africa? The evidence is limited. \textit{AIDS Behav.} 14(1), 17--24.

\bibitem{Foxman2006STD}
Foxman, B., Newman, M., Percha, B., Holmes, K. K. \& Aral, S. O. (2006) Measures of sexual partnerships: lengths, gaps, overlaps, and sexually transmitted infection. \textit{Sex. Transm. Dis.} 33(4), 209--214.

\bibitem{Doherty}
Doherty, I. A., Shiboski, S., Ellen, J. M., Adimora, A. A. \& Padian,N. S. (2006) Sexual bridging socially and over time: a simulation model exploring the relative effects of mixing and concurrency on viral sexually transmitted infection transmission. \textit{Sex Transm. Dis.} 33(6), 368--373.

\bibitem{Gurski}
Gurski, K. \& Hoffman, K. (2016) Influence of concurrency, partner choice, and viral suppression on racial disparity in the prevalence of HIV infected women. \textit{Math. Biosci.} 282, 91--108.

\bibitem{Eames}
Eames, K. T. D. \& Keeling, M. J. (2004) Monogamous networks and the spread of sexually transmitted diseases. \textit{Math. Biosci.} 189(2), 115--130.

\bibitem{Miller}
Miller, J. C. \& Slim, A. C. (2017) Saturation effects and the concurrency hypothesis: insights from an analytic model. \textit{PLoS ONE} 12, e0187938.

\bibitem{Onaga2017PRL}
Onaga, T., Gleeson, J. P. \& Masuda, N. (2017) Concurrency-induced transitions in epidemic dynamics on temporal networks. \textit{Phys. Rev. Lett.} 119(10), 108301.

\bibitem{Bauch2000PRSB}
Bauch, C. \& Rand, D. A. (2000) A moment closure model for sexually transmitted disease transmission through a concurrent partnership network. \textit{Proc. R. Soc. Lond. B} 267(1456), 2019--2027.

\bibitem{Pastor2001PRL}
Pastor-Satorras, R. \& Vespignani, A. (2001) Epidemic spreading in scale-free networks. \textit{Phys. Rev. Lett.} 86(14), 3200--3203.

\bibitem{Newman2}
Zhang, X., Moore, C. \& Newman, M. E. J. (2017) Random graph models for dynamic networks. \textit{Eur. Phys. J. B} 90(10), 200.

\bibitem{Clementi}
Clementi, A. E., Macci, C., Monti, A., Pasquale, F. \& Silvestri, R. (2008) Flooding time in edge-Markovian dynamic graphs. In: \textit{27th ACM Symposium on Principles of Distributed Computing}, pp. 213--222.

\bibitem{Elohim}
Fonseca dos Reis, E., Li, A. \& Masuda, N. (2020) Generative models of simultaneously heavy-tailed distributions of interevent times on nodes and edges. \textit{Phys. Rev. E} 102(5), 052303.

\bibitem{Ross}
Ross, S. M. (1995) \textit{Stochastic Processes}, 2nd ed., Wiley, Hoboken, NJ.

\bibitem{Barabasi1999Sci}
Barab\'{a}si, A. -L. \& Albert, R. (1999) Emergence of scaling in random networks. \textit{Science} 286(5439), 509--512.

\bibitem{Newman3}
Newman, M. E. J. (2006) Finding community structure in networks using the eigenvectors of matrices. \textit{Phys. Rev. E} 74(3), 036104.

\bibitem{Masuda2018SIAM} 
Masuda, N. \& Rocha, L. E. C. (2018) A Gillespie algorithm for non-markovain stochastic processes. \textit{SIAM Rev.} 60(1), 95--115.

\bibitem{Masuda2021arxiv}
Masuda, N. \& Vestergaard, C. L. (2023) \textit{Gillespie algorithms for stochastic multiagent dynamics in populations and networks}, Cambridge University Press, Cambridge.

\bibitem{Onaga2019book}
Onaga, T., Gleeson, J. P. \& Masuda, N. (2019) The effect of concurrency on epidemic threshold in time-varying networks. In: Holme, P. \& Saramäki, J. (editors), \textit{Temporal Network Theory}, Springer, Cham, pp. 253--267.

\bibitem{Karsai}
Karsai, M., Jo, H. \& Kaski, K. (2018) \textit{Bursty Human Dynamics}, Springer, Cham.

\bibitem{Sawers2013J.I.AIDS}
Sawers, L. (2013) Measuring and modelling concurrency. \textit{J. Int. AIDS Soc.} 16(1), 17431.



\bibitem{Farina2000}
Farina, L. \& Rinaldi, S. (2000) \textit{Positive Linear Systems — Theory and Applications}, John Wiley \& Sons, Inc., New York.
  
\bibitem{VanMieghem2009a}
Van Mieghem, P., Omic, J. \& Kooij, R. (2009) Virus spread in networks. \textit{IEEE/ACM Trans. Netw.} 17(1), 1--14.
  
\bibitem{Wilson1972-2010book}
Wilson, R. J. (2010) \textit{Introduction to Graph Theory}, 5th ed., Prentice Hall, Harlow.
  
\bibitem{Hashimoto1989AdvStudPureMath}
Hashimoto, K. (1989) Zeta functions of finite graphs and representations of $p$-adic groups. \textit{Adv. Stud. Pure Math.} 15, 211--280.

\bibitem{Alon2007CommContempMath}
Alon, N., Benjamini, I., Lubetzky, E. \& Sodin, S. (2007) Non-backtracking random walks mix faster. \textit{Commun. Contemp. Math.} 9(04), 585--603.
  
\bibitem{Masuda2020ImaJApplMath}
Masuda, N. \& Ogura, M. (2020) Analysis of the susceptible-infected-susceptible epidemic dynamics in networks via the non-backtracking matrix. \textit{IMA J. Appl. Math.} 85(2), 214--230.

\bibitem{Ganesh2005}
Ganesh, A., Massoulié, L. \& Towsley, D. (2005) The effect of network topology on the spread of epidemics. In: Proc. IEEE 24th Annual Joint Conference of the IEEE Computer and Communications Societies (INFOCOM’05), pp. 1455--1466.

\bibitem{Chakrabarti2008}
Chakrabarti, D., Wang, Y., Wang, C., Leskovec, J. \& Faloutsos, C. (2008) Epidemic thresholds in real networks. \textit{ACM Trans. Inf. Syst. Secur.} 10(4), 1--26.
  
\bibitem{Preciado2014}
Preciado, V. M., Zargham, M., Enyioha, C., Jadbabaie, A. \& Pappas, G. J. (2014) Optimal resource allocation for network protection against spreading processes. \textit{IEEE Trans. Control Netw. Syst.} 1(1), 99--108.
  
\bibitem{Ogura2018SystControlLett}
Ogura, M. \& Preciado, V. M. (2018) Second-order moment-closure for tighter epidemic thresholds. \textit{Syst. Control Lett.} 113, 59--64.
  
\bibitem{Horn1990}
Horn, R. \& Johnson, C. (2012) \textit{Matrix Analysis}, 2nd ed., Cambridge University Press, Cambridge.

\bibitem{Damm2003LinAlgItsAppl}
Damm, T. \& Hinrichsen, D. (2003) Newton’s method for concave operators with resolvent positive derivatives in ordered Banach spaces. \textit{Linear Algebra Appl.} 363, 43--64.

































\end{thebibliography}
\end{document}